\documentclass[pdflatex,11pt]{amsart}

\usepackage{amsthm, amsmath, amssymb, amsfonts}

\usepackage[OT4]{fontenc}
\usepackage[utf8]{inputenc}
\usepackage[english]{babel}

\usepackage{bbm}

\usepackage[inline]{enumitem}

\usepackage{float}
\usepackage{algorithm}
\usepackage{algpseudocode}
\restylefloat{table}
\usepackage{lipsum}

\usepackage[colorlinks=true, pdfstartview=FitV, linkcolor=blue,
            citecolor=blue, urlcolor=blue]{hyperref}
\usepackage[usenames]{color} 
\usepackage[round]{natbib}

\usepackage{graphicx}
\graphicspath{{Images/}}
\usepackage[font=sl,labelfont=bf,width=\textwidth]{caption}
\usepackage{color}

\newcommand{\wo}{\mathbb{E}}
\DeclareMathOperator{\cov}{Cov}

\usepackage{multirow}
\usepackage{amsthm}
\newtheorem{theorem}{Theorem}

\newtheorem{proposition}[theorem]{Proposition}

\newtheorem{remark}[theorem]{Remark}

\DeclareMathOperator{\Var}{Var}   
\def\bP{\mathbb{P}}
\def\bE{\mathbb{E}}
\def\bR{\mathbb{R}}
\def\bN{\mathbb{N}}
\def\cF{\mathcal{F}}
\usepackage{bbm}

\usepackage{amssymb}
\usepackage{placeins} 

\usepackage[foot]{amsaddr}

\usepackage[margin=1in]{geometry} 

\usepackage{setspace}
\begin{document} 











\title[Statistical applications of the 20/60/20 rule]{Statistical applications  of the 20/60/20 rule in risk management and portfolio optimization}

\author{Kewin P\k{a}czek$^{\ast}$}
\author{Damian Jelito$^{\ast}$}
\author{Marcin Pitera$^{\ast}$}
\author{Agnieszka Wy\l{}oma\'{n}ska$^{\dagger}$}
\address{$^{\ast}$Institute of Mathematics, Jagiellonian University, S. {\L}ojasiewicza 6, 30-348 Krak{\'o}w, Poland}
\address{$^{\dagger}$Faculty of Pure and Applied Mathematics, Hugo Steinhaus Center, Wroc{\l}aw University of Science and Technology, Wyspia{\'n}skiego 27, 50-370 Wroc{\l}aw, Poland}
\email{kewin.paczek@im.uj.edu.pl}
\email{damian.jelito@uj.edu.pl}
\email{marcin.pitera@uj.edu.pl}
\email{agnieszka.wylomanska@pwr.edu.pl}

\vspace{-1cm}
\begin{abstract}
This paper explores the applications of the 20/60/20 rule—a heuristic method that segments data into top-performing, average-performing, and underperforming groups—in mathematical finance. We review the statistical foundations of this rule and demonstrate its usefulness in risk management and portfolio optimization. Our study highlights three key applications. First, we apply the rule to stock market data, showing that it enables effective population clustering. Second, we introduce a novel, easy-to-implement method for extracting heavy-tail characteristics in risk management. Third, we integrate spatial reasoning based on the 20/60/20 rule into portfolio optimization, enhancing robustness and improving performance. To support our findings, we develop a new measure for quantifying tail heaviness and employ conditional statistics to reconstruct the unconditional distribution from the core data segment. This reconstructed distribution is tested on real financial data to evaluate whether the 20/60/20 segmentation effectively balances capturing extreme risks with maintaining the stability of central returns. Our results offer insights into financial data behavior under heavy-tailed conditions and demonstrate the potential of the 20/60/20 rule as a complementary tool for decision-making in finance.
\vspace{0.2cm}

\end{abstract}
\maketitle

\linespread{1.25}

\section{Introduction}
The 20/60/20 rule is a heuristic approach that offers a flexible framework for analyzing and modeling different types of data by segmenting observations according to the values of a given benchmark; see \cite{206020Rule1999, 206020Rule2009} for details. In a nutshell, one divides observations into three regions: 20\% in each of the extreme tails and 60\% in the central portion and deals with each subset separately; see \cite{BlackSwan,taleb_black_2008}. This segmentation reflects the varying statistical properties observed in different parts of the data and facilitates region-specific treatment, see Section~\ref{section:rule20/60/20} for details. In this article, we check if this reasoning could be efficiently applied to financial datasets in order to enhance existing quantitative frameworks. Specifically, we implement the 20/60/20 framework in finance, aiming to develop a more adaptive and refined approach to risk management and portfolio optimization.

Mathematically, the 20/60/20 rule can be understood within the framework of conditional statistics, which draws connections to the central limit theorem and datasets composed of independent observations driven by similar underlying innovations. This foundation allows for the analysis of conditional variance within each subset (20\%, 60\%, 20\%), enabling a deeper understanding of distributional properties. Conditional moments, in particular, have proven to be valuable for assessing goodness of fit, such as testing normality (see \cite{JelPit2021, Wozny2024GaussianDS}) or evaluating distributions within the general class of the $\alpha$-stable distributions (\cite{PitCheWyl2021, PacJelPitWyl2023}). Additionally, they facilitate robust parameter estimation in heavy-tailed environments, as demonstrated in \cite{PacJelPitWyl2022}. By leveraging these statistical tools, the 20/60/20 rule provides a structured approach to analyzing variability and patterns within segmented datasets.

In this article, we check if the 20/60/20 rule might be applied to finance into three different areas: direct visual inspection tools, distribution characteristic extraction, and spatial-based control decision-making. For each of those applications, we presented particular examples: visual clustering, measuring tail thickness, and portfolio optimization, respectively. 

First, our research suggests that the 20/60/20 rule can be useful for clustering financial data by dividing the observations into "tails" and the middle part. This approach bridges the gap between classical methods that treat whole samples homogeneously, and complex econometrics models, offering a practical tool for modern spatial-conditioning risk  management. Segmentation ensures that the heavy tails are not ignored but treated separately, enabling a clearer distinction between stable returns in the central portion and extreme risk in the tails. Surprisingly, we noticed that in many cases, the 60\% of middle observations of the stock market tend to exhibit normal behavior. What is more, this refers to both stock and index data, and is observed on different time horizons, so that it can be seen as a universal phenomenon.

Second, this approach introduces a novel measure of tail thickness for one-dimensional random variables. By calculating the second moment (or variance) conditioned on the subset of the middle 60\% of observations, we can extrapolate the conditional variance to estimate the unconditional variance, assuming homogeneous behavior across the central and tail subsets. Comparing this extrapolated conditional variance with the unconditional variance computed from the entire sample provides a meaningful measure of tail thickness. This comparison can reveal the presence of unexpected extreme tails or deviations from the assumed distribution, offering a valuable alternative to traditional measures like kurtosis or other dispersion metrics, see \cite{SHIN2025110325,joanes1998comparison,Staudte03042017} for detalis. Such a method highlights the potential for identifying subtle distributional characteristics that might otherwise be overlooked.

Third, we show that the 20/60/20 rule can be effectively applied to enhance or benchmark portfolio optimization models. Our work introduces a novel estimator for the covariance matrix, specifically tailored to the central portion of the data. This estimator modifies the classical plug-in Markowitz approach by focusing on the middle 60\% of returns, where risk can be modeled more reliably without the distortions introduced by extreme tail behavior. The rationale behind this approach is that while one cannot control or predict extreme events, it is more practical to build a model that performs well in standard market conditions. 

It is important to note that this is a toy exercise, designed to demonstrate the potential of applying the 20/60/20 rule to mathematical finance. While our examples show that this approach is plausible and offers promising insights, further refinement is needed to make it applicable to real-world strategies. For instance, future work could incorporate specific measures to control for extreme tail behavior, enhancing the model's robustness in more volatile environments.

This paper is organized as follows. In Section~\ref{section:rule20/60/20} we formulate the 20/60/20 rule and discuss its various applications with the special attention to financial markets. In Section~\ref{section:math} we present the mathematical foundation of the 20/60/20 rule and the related concepts of quantile conditional moments. In Section~\ref{s:statistical_setup} we extend the theoretical results from Section~\ref{section:math} to the sample data setting and define the corresponding estimators.
In Section~\ref{section:empirical} we evaluate the potential utility of newly defined statistics by analyzing their performance for financial data sets using three different examples. In Section~\ref{s:conclusions}, we conclude the paper with a comment on possible extensions of the proposed approaches.

\section{The 20/60/20 rule and its potential applications to finance}\label{section:rule20/60/20}
In resource allocation and management, various rules are used to divide resources into subsets. Common examples include the 20/60/20 rule, the 30/30/40 rule, and the 20/80 rule, see e.g. Section 3 in \cite{covey1995first}. Each of these rules is based on certain implicit assumptions about the data. For instance, as discussed by~\cite{strand2020,digeronimo2020}, the 30/30/40 rule (or the similar  20/60/20 rule) often assumes a more balanced and independent distribution of resources or tasks. In contrast, the 20/80 rule aligns with Pareto optimality principles (see \cite{Pareto1896} for details), emphasizing selective attention on top-performing elements to drive significant gains; we refer to \cite{Marshall2013,Newman01092005} for further insights. However, these are primarily rules of thumb, where the specific percentages are less important and often arbitrary. What truly matters is the underlying principle: different subsets require different approaches, and tailoring strategies accordingly can lead to better overall outcomes. The key takeaway is not the exact split but the recognition that applying distinct methods to different groups can enhance results. In this article, we focus on the 20/60/20 rule due to its mathematical justification and our strong belief in its applicability to finance.

The 20/60/20 rule is a widely recognized concept in management science. It is used to segment groups based on their performance or impact; see \cite{Kamiya2005WorkerAR, 206020Rule2009} and references therein. According to this rule, when a predefined benchmark can measure effectiveness, a population can be divided into three distinct groups based on their level of performance. 
Following \cite{206020Rule1999}, we can apply this rule in the human resource management setting, and divide employees into three groups:
\begin{itemize}

     \item Positive group: They are the best performers with respect to the benchmark criterion value. This can relate to the top 20\% of employees whose productivity, engagement, or effectiveness markedly surpasses the norm. These members drive positive change, innovation, and efficiency and they are often considered key assets in achieving organizational goals.

    \item Neutral group: Representing the middle 60\%, this group demonstrates average performance. As the largest segment of the population, they tend to exhibit neutral or stable behavior. In the workplace, these individuals are generally consistent but neither highly productive nor harmful. Their impact on overall outcomes can vary, often requiring support or guidance to foster positive engagement and enhance productivity.

    \item Negative group: The bottom 20\% generally consists of individuals or entities that have a detrimental effect on the overall objective. In the context of workplace productivity, this might include employees who consistently underperform or engage in counterproductive behaviors. The negative group's impact often results in reduced efficiency, decreased morale, and, in some cases, a need for corrective action or realignment.
\end{itemize}

Although not inherently mathematical, the 20/60/20 framework, similar to, mentioned before, the 20/80 rule, provides a pragmatic tool for managing complexity. The numerical splits, whether derived from normality or tailored to specific contexts, underline the importance of differentiated strategies for optimizing outcomes. However, further investigation is needed to validate the effectiveness of such split-control approaches.
The popularity of the 20/60/20 heuristic stems from its simplicity and broad applicability across various domains. It serves as a useful benchmark for analyzing observed patterns of human behavior and performance, particularly within organizational contexts. For example, from a psychological perspective, the 20/60/20 division can be linked with natural variations in human traits and behavioral tendencies; see \cite{RyanMotivation}. When assessing populations that follow a multivariate normal distribution, such as in psychometric or workforce assessments, the clustering effect of high, average, and low performers often emerges naturally; see \cite{McGreagor1960} for details. This division aligns with concepts in social psychology and behavioral economics, where populations rarely conform to purely homogeneous standards and instead exhibit diverse levels of motivation, engagement, and productivity. In statistics, the 20/60/20 rule can be also used to efficiently test the normality, see \cite{JelPit2021} for more details.
The 20/60/20 rule, though simplistic, serves as a starting point for nuanced decision-making strategies across various domains. The key aspect of the rule is not the splitting itself, but the fact that different approaches should be applied to different groups, which might lead to better results.

The 20/60/20 rule, while not widely recognized in quantitative finance, has potential applications across various financial domains. Despite its heuristic nature and lack of rigorous theoretical foundation, it is often used as a practical guideline. For instance, in budgeting, it suggests allocating 60\% of income to living expenses, 20\% to savings, and 20\% to discretionary spending; see \cite{206020Rule2024} for details. Beyond budgeting, we propose that the 20/60/20 rule can be effectively applied to other areas of finance, serving as a versatile tool for data segmentation, portfolio management, and trading strategies.

One promising application of the 20/60/20 rule is linked to a clustering mechanism, where data is segmented into three distinct categories based on a predefined benchmark. Probabilistic clustering models, which are foundational to modern machine learning, provide a robust framework for categorizing data according to their statistical properties. These models are particularly valuable in finance, where datasets often exhibit diverse distributions and high levels of uncertainty. Techniques such as Gaussian Mixture Models (GMMs) and Hidden Markov Models (HMMs) are widely used to identify different market regimes in financial time series data; please see \cite{Hamilton1990, Hastie2009,Alexander2009} for details. Additionally, probabilistic clustering aids in anomaly detection within transaction data, supporting fraud detection efforts, we refer to \cite{Bolton2002} for details. By integrating data-driven insights with the 20/60/20 rule, these methods can help design financial strategies tailored to specific risk profiles and market conditions. For further details on clustering techniques, refer to Chapter 3 in \cite{Clustering2014}.

In portfolio management, the 20/60/20 rule can be applied to prioritize investments. For example, investors could focus on the top 20\% of high-performing assets, actively monitor the middle 60\% for potential growth opportunities, and consider divesting or re-evaluating the bottom 20\% of underperforming investments. Similarly, in trading system algorithms, the rule can be used to identify profitable investment opportunities. For instance, it might be advantageous to invest in stocks where the variance in the tails is relatively small compared to the variance in the middle 60\% of observations. Additionally, by analyzing daily trading volume and applying the 20/60/20 rule, stocks can be categorized into three groups, providing insights into which stocks are more popular or liquid.

In this paper, we present three simple yet effective approaches to applying the 20/60/20 rule in finance. While these approaches demonstrate its utility, we firmly believe that the rule has a broader range of potential applications, and we hope this work inspires further exploration of its use in financial modeling and decision-making.

\section{Mathematical setup}\label{section:math}

In this section, we introduce a formal mathematical setup and provide more details on the structure of the conditional covariance matrix for multivariate normal distributions. The key purpose of this section is to show that the 20/60/20 rule induces a specific form of population balance under the assumption of multivariate normality, and to demonstrate how to expand conditional variance to an unconditional setup in order to measure the divergence from the normality-induced balance.

\subsection{Conditional covariance matrices}

Let $(\Omega,\cF,\bP)$ be a probability space. For a fixed $d\in\bN$, let $X=(X_1,\dots,X_d)$ be an absolutely continuous $d$-dimensional vector with the square-integrable coordinates, the cumulative distribution function (CDF) denoted by $F_X$, and the probability density function (PDF) denoted by $f_X$. We fix a strictly positive loading factor vector $a:=(a_1,\ldots,a_d)\in\bR_+^d$ and use
\begin{equation}\label{eq:Y}
Y:=\sum_{k=1}^d a_kX_k
\end{equation}
to denote the linear combination of margins of $X$; we often refer to $Y$ as a {\it benchmark} for $X$ and use $F_Y$, $f_Y$ to denote CDF and  PDF of $Y$, respectively. We also use $(\cdot)^T$ to denote vector transposition. Note that we can write $Y=\langle a,X\rangle$ with the Euclidean scalar product $\langle\cdot,\cdot\rangle$. Next, we fix a pair of quantiles $p,q\in (0,1)$, where $p<q$, and define the benchmark quantile conditional set
\begin{equation}\label{eq:B.set}
B := [F^{-1}_Y(p),F^{-1}_Y(q)],
\end{equation}
where $F^{-1}_Y$ is the inverse of the CDF of $Y$. For $i,j=1,\ldots, d$, the corresponding 
{\it quantile conditional variance} and {\it quantile conditional covariance} are given by
\begin{align}
    \Var(X_i|Y\in B) &:= \bE[(X_i-\bE[X_i|Y\in B])^2|Y\in B],\label{eq:cond_var}\\
    \cov(X_i,X_j|Y \in B) &:= \bE\left[(X_i-\bE[X_i|Y\in B])  (X_j-\bE[X_j|Y\in B])|Y \in B\right].\label{eq:cond_cov}
\end{align} 
Note that all objects are well defined and finite since $\bP[Y \in B]>0$, $p,q\in (0,1)$, and $X$ is in $L^2$. Also, with a slight abuse of notation, for a continuous, $L^2$ random variable $Z$ we use 
$\Var(Z|Y \in B)$ and $\cov(X_i,Z|Y \in B)$ to denote conditional 
variance and covariance of $Z$ on the set $\{Y\in B\}$. For brevity, we also use vector notation
\begin{align*}
    \bE[X|Y\in B] &:=(\bE[X_1|Y\in B],\ldots,\bE[X_d|Y\in B])\\
     \Var(X|Y \in B) &:=(\Var(X_1|Y \in B),\ldots,\Var(X_d|Y \in B)),\\
    \cov(X,Z|Y \in B) &:= (\cov(X_1,Z| Y \in B),\ldots,\cov(X_d,Z | Y \in B),\\
 \cov(X,Z) &:= (\cov(X_1,Z),\dots,\cov(X_d,Z)),
 \end{align*}
and use $\cov(X| Y\in B)$ to denote the (conditional) variance-covariance matrix of $X$.
As usual, we use $\Phi$ and $\phi$ to denote CDF and PDF of the standard normal random variable. We conclude this subsection by introducing the notation for the conditional variance of the standard normal random variable. More specifically, for quantile splits $p,q\in (0,1)$ we define
\begin{equation}\label{eq:s}
s(p,q) := \left(\frac{\Phi^{-1}(p)\phi(\Phi^{-1}(p))-\Phi^{-1}(q)\phi(\Phi^{-1}(q)) }{q-p} \right) - \left( \frac{(\phi(\Phi^{-1}(p))-\phi(\Phi^{-1}(q)))^2}{(q-p)^2}\right) +1.
\end{equation}
Note that, assuming $Z\sim \mathcal{N}(0,1)$, we get $s(p,q)=\Var(Z|Z \in B)$; see Section 13.10.1 in~\cite{JohKotBal1994} for the proof. Finally, please note that the assumption of square-integrability of $X$ is introduced for simplicity and guarantees that the conditional first and second moments are well-defined. Still, many of the results could be generalized as the conditioning often leads to the effective boundedness of the underlying random variables; see \cite{PitCheWyl2021}.

\subsection{Mathematical justification of 20/60/20 rule}\label{s:intro.math}
In this subsection, we recall the result that provides a mathematical justification for the 20/60/20 rule.  Throughout this subsection, we assume that $X$ follows multivariate $d$-dimensional normal distribution with mean $\mu\in\bR^{d}$ and covariance matrix $\Sigma$. For brevity, we use the notation $X \sim \mathcal{N}_d(\mu, \Sigma)$. Therefore, the benchmark random variable $Y$ is a one-dimensional normal random variable. 
As in the previous subsection, we condition $X$ based on the values of $Y$ with the help of the following fixed conditioning sets
\begin{equation}\label{eq:A1A2A3}
B_1:=\left(-\infty,F_Y^{-1}(\tilde q)\right],\qquad B_2:=\left(F_Y^{-1}(\tilde q),F^{-1}(1-\tilde q)\right]\qquad B_3:=\left(F_Y^{-1}(1-\tilde q),\infty\right],
\end{equation}
where $\tilde{q}:=\Phi(\tilde{x})\approx 0.19808$,  and $\tilde{x}$ is a unique positive solution to the equation 
\begin{equation*}
 -x \Phi(x)-\phi(x)(1-~2 \Phi(x))=~0,   
\end{equation*}
 see Theorem 1 in \cite{JawPit2014} for details. The sets $B_1$, $B_2$, and $B_3$ correspond to approximately top $20\%$, middle $60\%$, and lower $20\%$ values of $Y$. In Theorem 3.1 in~\cite{JawPit2015}, it has been observed that the sets $\{Y\in B_1\}$, $\{Y\in B_2\}$, and $\{Y\in B_3\}$ creates a spatial balance of a random vector $X$ with respect to the conditional covariance matrix. The 20/60/20 split is, in fact, a covariance equilibrium invariant with respect to the choice of all underlying parameters, i.e. $d$, $\mu$, $\Sigma$, and $a$. This is formalized in the following theorem. 
\begin{theorem}[\cite{JawPit2015}]\label{th:JawPit}
Let $X\sim \mathcal{N}_d(\mu,\Sigma)$, $a\in\bR^d_{+}$, and $Y=\langle a,X\rangle$. Then, we have
\begin{equation}\label{eq:206020}
\cov(X\mid Y\in B_1)=\cov(X\mid Y\in B_2)=\cov(X\mid Y \in B_3),
\end{equation}
where $B_1, B_2, B_3$ are given in $\eqref{eq:A1A2A3}$. 
\end{theorem}
The 20/60/20 rule has been identified as a versatile tool that could be applied in developing estimation and testing frameworks for statistical hypotheses, see e.g. \cite{JelPit2021, Wozny2024GaussianDS}. These applications underscore the broad relevance and utility of the 20/60/20 rule, demonstrating its potential to enhance both theoretical insights and practical implementations of statistical methods. In this paper, we expand the scope of its applications by exploring its use in three key areas: as a clustering tool, as a measure of heavy-tailed behavior, and as a framework for refining portfolio optimization strategies. 

\subsection{Central set covariance expansion}
In this subsection, we show how to recover the unconditional covariance from its conditional version. The key result of this section, Theorem~\ref{th:OurTh},  could be seen as a dual statement to the one presented in Theorem 4.1 in~\cite{JawPit2017}, where the formula for the conditional covariance has been provided.
\begin{theorem}\label{th:OurTh} 
Let $X\sim \mathcal{N}_d(\mu,\Sigma)$, $a\in\bR_+^d$, and $Y=\langle a,X\rangle$. Then, for any $0<p<q<1$, and the related set $B$  given in~\eqref{eq:B.set}, we get 
\begin{eqnarray}\label{eq:cond.cov.est}
         \cov(X) = \cov(X|Y \in B)- (1 - \tfrac{1}{s(p,q)})\Var(Y|Y \in B)\left(\frac{\cov(X,Y|Y\in B)}{\Var(Y|Y \in B)}\right) \left(\frac{\cov(X,Y| Y \in B)}{\Var(Y| Y\in B)}\right)^T.
\end{eqnarray}
\end{theorem}
\begin{proof}
The proof of Theorem~\ref{th:OurTh} is based on the arguments used in the proof of Theorem 4.1 in ~\cite{JawPit2017}. 
Let $\beta=(\beta_1,\ldots,\beta_d) := \frac{\cov(X,Y)}{\Var(Y)}$ and $\varepsilon=(\varepsilon_1,\ldots,\varepsilon_d):=\beta Y-X$. Since $(X,Y,\varepsilon)$ is a multivariate random vector, the simple linear regression formula $X=\beta Y+\varepsilon$ is satisfied. In particular, for any $i \in \{1,\ldots,d\}$, we have
\begin{equation}\label{eq:cov.eps.Y}
\cov(\varepsilon_i,Y) =0,
\end{equation}
and the uncorrelation implies independence due to the multivariate normality of vector $(\varepsilon_i,Y)$. Now, let us show that, for $i,j=1,\ldots d$, we get
\begin{equation}\label{eq:th.1}
\cov(X_i,X_j) = \cov(X_i,X_j|Y\in B) - \beta_i\beta_j (\Var(Y|Y\in B) - \Var(Y)).
\end{equation}
Note that \eqref{eq:th.1} could be rewritten in a matrix form as 
\begin{equation*}
\cov(X) =  \cov(X|Y\in B) - (\Var(Y|Y\in B) - \Var(Y))\beta\beta^{T}.
\end{equation*}
For a fixed $i,j\in \{1,\ldots,d\}$, we get
\begin{equation}\label{eq:covXX}
\cov(X_i,X_j)=\cov(\beta_iY - \varepsilon_i,\beta_jY - \varepsilon_j) =\beta_i\beta_j\Var(Y) + \cov(\varepsilon_i,\varepsilon_j) - \beta_i\cov(Y,\varepsilon_j) - \beta_j\cov(\varepsilon_i,Y). 
\end{equation}
Due to Equation~\eqref{eq:cov.eps.Y}, the last two components are equal to zero. Therefore, because of the independence $Y$ and $\varepsilon_i$, we can write
\begin{align}\label{eq:conXXFinal}
\cov(X_i,X_j) &=\beta_i\beta_j\Var(Y) + \cov(\varepsilon_i,\varepsilon_j)
= \beta_i\beta_j\Var(Y) + \cov(\varepsilon_i,\varepsilon_j|Y\in B) \nonumber \\
&=\beta_i\beta_j\Var(Y) + \cov(\beta_iY - X_i,\beta_jY - X_j|Y\in B) \nonumber \\
&= \cov(X_i,X_j|Y\in B)+ \beta_i\beta_j \left(\Var(Y) + \Var(Y|Y\in B) \right) \nonumber\\
&\phantom{=}- \beta_i\cov(Y,X_j | Y \in B)-\beta_j\cov(X_i,Y|Y\in B).
\end{align}
Next, using again the independence of $Y$ and $\varepsilon_i$, we have
\begin{equation}\label{eq:th1:2}
\beta_i= \beta_i \frac{\cov(Y,Y|Y \in B)}{\Var(Y|Y \in B)} + \frac{\cov(\varepsilon_i,Y|Y\in B)}{\Var(Y|Y \in B)} = \frac{\cov(\beta_i Y+\varepsilon_i,Y|Y \in B)}{\Var(Y|Y\in B)} = \frac{\cov(X_i,Y|Y\in B)}{\Var(Y|Y\in B)}.
\end{equation}
Consequently, plugging \eqref{eq:th1:2} into \eqref{eq:conXXFinal} we get
\begin{align*}
\cov(X_i,X_j)
 &= \cov(X_i,X_j|Y\in B) + \beta_i\beta_j \left(\Var(Y) + \Var(Y|Y\in B) \right) \\
 &\phantom{=} -2\frac{\cov(X_i,Y|Y\in B)\cov(X_j,Y|Y \in B)}{\Var(Y|Y\in B)} \\
 &= \cov(X_i,X_j|Y\in B) + \beta_i\beta_j \left(\Var(Y) + \Var(Y|Y\in B) \right)  - 2 \beta_i\beta_j \Var(Y|Y\in B) \\
 &= \cov(X_i,X_j|Y\in B) -\left( \Var(Y|Y\in B) -\Var(Y)\right)\beta_i\beta_j.
\end{align*}
Finally, noting that $(\Var(Y|Y\in B) - \Var(Y)) = (1-\frac{1}{s(p,q)}) \Var(Y|Y\in B)$ where $s(p,q)$ is given by~\eqref{eq:s}, we conclude the proof.
\end{proof}
Theorem~\ref{th:OurTh} is, in fact, true for any generic measurable set $B$, with $s(p,q)$ replaced by the appropriate normalizing factor. Still, to streamline the presentation, we decide to state the result and the proof only for quanitle conditioning sets.    In particular, Equation~\eqref{eq:th1:2} shows that the ratio of variance to covariance (that is, regression coefficient) is proportional for any conditioning set $B$.

Theorem \ref{th:OurTh} demonstrates that for the normal distribution it is possible to recover the unconditional covariance using the middle 60\% of the distribution. Thus, it can be used as a base for a measure of heavy-tailed behavior by comparing the unconditional distribution of a given random variable with the recovered normal distribution. 
To evaluate the practicality of this approach, it is necessary to establish a formal statistical framework along with the corresponding estimators and test statistics.


\section{Statistical setup}\label{s:statistical_setup}

 In this section, we introduce the statistical setup and show how to estimate the statistics defined in Equations \eqref{eq:cond_var}, \eqref{eq:cond_cov}, and \eqref{eq:cond.cov.est}. It is important to note that in this section we do not assume the normality of the random vector $X$. For a fixed sample size $n\in\bN\setminus \{0\}$, we use
\begin{equation}\label{eq:sample}
\mathbf{X}:=(X_{1,k},\ldots,X_{d,k})_{k=1}^n
\end{equation}
to denote an $n$-element independent identically distributed (i.i.d.) sample from the distribution of $X$.  For brevity, for $i=1,\ldots,d$, we also use $\mathbf{X}_i$ to denote $i$th margin sub-sample of $\mathbf{X}$, that is, we set $\mathbf{X}_i:=(X_{i,k})_{k=1}^n$. Similarly, we set
\begin{equation} \label{eq:sample_Y}
    \mathbf{Y}:=(Y_k)_{k=1}^{n},
\end{equation}
where $Y_k:=\sum_{i=1}^da_i X_{i,k}$, to denote a benchmark sample from $Y$ with a fixed vector $a$.


Given a sample $\mathbf{X}$ and the induced benchmark sample $\mathbf{Y}$, we define the empirical version of the conditioning set $B$ by setting 
\begin{equation}\label{eq:set_B.hat}
\hat{B} :=[Y_{([np]+1)},Y_{([nq])}],
\end{equation}
where we use ”(k)” to denote the index corresponding to $kth$ order statistic of the sample $\mathbf{Y}$ and $[b]:=\sup\{k\in \mathbb{Z}\colon k\leq b\}$ to denote the integer part of $b\in \mathbb{R}$.
Let us note that in such a case, the intervals are random, that is, the value $\hat{B}$ depends on a specific realization of the underlying sample.
Next, using \eqref{eq:set_B.hat}, we define the sample conditional mean, along with the sample equivalents of \eqref{eq:cond_var} and \eqref{eq:cond_cov}. Namely, we set
\begin{align}
    \textstyle \hat\mu_B(\mathbf{X_i})&\textstyle :=\frac{1}{\sum_{k=1}^n 1_{\hat{B}}(Y_k) } \sum_{k=1}^n X_{i,k} 1_{\hat{B}}(Y_k), \label{eq:cond_exp_emp}\\
    \textstyle\hat\sigma^2_B(\mathbf{X_i}) &\textstyle:= \frac{1}{\sum_{k=1}^n 1_{\hat{B}}(Y_k)} \sum_{k=1}^n (X_{i,k} - \hat\mu_B(\mathbf{X_i}))^2  1_{\hat{B}}(Y_k),\label{eq:cond_var_emp}\\
        \textstyle\widehat\cov_B(\mathbf{X_i},\mathbf{X_j}) &\textstyle:= \frac{1}{\sum_{k=1}^n 1_{\hat{B}}(Y_{k})} \sum_{k=1}^n (X_{i,k} - \hat\mu_B(\mathbf{X_i})) (X_{j,k} - \hat\mu_B(\mathbf{X_j})) 1_{\hat{B}}(Y_k).\label{eq:cond_cov_emp_X}
\end{align}
As in Section~\ref{s:intro.math}, we also define vector estimators and respective conditional statistics on $\mathbf{Y}$. More specifically, we set
\begin{align*}
\textstyle\hat\mu_B(\mathbf{Y}) & \textstyle:=\frac{1}{\sum_{k=1}^n 1_{\hat{B}}(Y_k) } \sum_{k=1}^n Y_{k} 1_{\hat{B}}(Y_k),\\
\textstyle\hat\sigma^2_B(\mathbf{X}) &\textstyle =(\hat\sigma^2_B(\mathbf{X}_1),\ldots, \hat\sigma^2_B(\mathbf{X}_d)), \\
\textstyle \hat\sigma^2_B(\mathbf{Y})&:=\textstyle \sum_{i,j = 1}^d a_ia_j \widehat\cov_B(\mathbf{X_i},\mathbf{X_j}) \\
\textstyle \widehat\cov_B(\mathbf{X_i},\mathbf{Y})&\textstyle := \frac{1}{\sum_{k=1}^n 1_{\hat{B}}(Y_k)} \sum_{k=1}^n (X_{i,k} - \hat\mu_B(\mathbf{X_i})) (Y_k - \hat\mu_B(\mathbf{Y})) 1_{\hat{B}}(Y_k).
\end{align*} 

The estimators presented in equations \eqref{eq:cond_exp_emp}--\eqref{eq:cond_cov_emp_X} are constructed by selecting a sub-sample based on empirical quantiles and then applying standard estimators to compute the desired statistic. This can be seen by observing that
\begin{align*}
    \textstyle\hat\mu_{B}(\mathbf{X_i})&\textstyle =\frac{1}{[nq] - [np] }  \sum_{k=[np]+1}^{[nq]} X_{i,(k)}, \\
    \textstyle\hat\sigma^2_B(\mathbf{X_i}) &\textstyle = \frac{1}{[nq] - [np] }  \sum_{k=[np]+1}^{[nq]} (X_{i,(k)} - \hat{\mu}_{B}(\mathbf{X_i}))^2,\\\textstyle\widehat\cov_B(\mathbf{X_i},\mathbf{X_j}) &\textstyle =  \frac{1}{[nq] - [np] }  \sum_{k=[np]+1}^{[nq]} (X_{i,(k)} - \hat{\mu}_{B}(\mathbf{X_i})) (X_{j,(k)} - \hat{\mu}_{B}(\mathbf{X_j})),
\end{align*}
 where $X_{j,(i)}$ is ranked
according to the values of the benchmarked sample $(Y_{i})_{i=1}^n$; see Appendix A in~\cite{Wozny2024GaussianDS} for details. Note that by setting $p=0$ and $q=1$ we recover standard  (unconditional) estimators of mean, variance, and covariance that are given by
\begin{align}
\textstyle\hat\mu(\mathbf{X_i})& \textstyle :=\frac{1}{n}  \sum_{k=1}^{n} X_{i,k}, \label{eq.standard.mean} \\
\textstyle\hat\sigma^2(\mathbf{X_i}) &\textstyle:= \frac{1}{n}  \sum_{k=1}^{n} (X_{i,k} -  \hat\mu(\mathbf{X_i}))^2,\label{eq:standard.variance.est}\\   
\textstyle\hat\sigma^2(\mathbf{Y}) &\textstyle:= \frac{1}{n}  \sum_{k=1}^{n} (Y_{k} -  \hat\mu(\mathbf{Y}))^2,\label{eq:standard.variance.Y.est}\\
\textstyle\widehat\cov(\mathbf{X_i}, \mathbf{X_j}) &\textstyle := \frac{1}{n }  \sum_{k=1}^{n} (X_{i,k} - \hat\mu(\mathbf{X_i})) (X_{j,k} - \hat\mu(\mathbf{X_j})).\label{eq:standard.covariance.est}
\end{align} 
Next, we can use Theorem~\ref{th:OurTh} to define alternative estimators of unconditional variance and covariance for multivariate normal distributions that are given by
\begin{align}
   \textstyle \bar\sigma^2(\mathbf{X_i})& \textstyle:=\hat\sigma^2_B(\mathbf{X_i}) - (1 - \frac{1}{s(p,q)}) \left(\frac{\textstyle\widehat\cov^2_B(\mathbf{X_i},\mathbf{Y})}{\textstyle\hat\sigma^2_B(\mathbf{Y})} \right), \label{eq:variance_estimator} \\
    \textstyle\overline{\cov}(\mathbf{X_i}, \mathbf{X_j}) &\textstyle:= \textstyle\widehat\cov_B(\mathbf{X_i},\mathbf{X_j}) - (1 - \frac{1}{s(p,q)}) \textstyle\hat\sigma^2_B(\mathbf{Y}) \left( \frac{\textstyle\widehat\cov_B(\mathbf{X_i},\mathbf{Y})}{\textstyle\hat\sigma^2_B(\mathbf{Y})}\right) \left(\frac{\textstyle\widehat\cov_B(\mathbf{X_j},\mathbf{Y})}{\textstyle\hat\sigma^2_B(\mathbf{Y})} \right). \label{eq:con_covariance_estimator} \\
   \textstyle \bar\sigma^2(\mathbf{Y})&: \textstyle= \frac{\textstyle\hat\sigma^2_B(\mathbf{Y})}{s(p,q)} . \label{eq:variance_estimator.Y}
\end{align}  
 Note that these estimators are based on Equation~\eqref{eq:cond.cov.est} and tacitly assume that a sample follows a multivariate normal distribution.
 \begin{remark}
     Estimators defined in Equations       \eqref{eq:cond_exp_emp}-\eqref{eq:cond_cov_emp_X} are consistent. Furthermore, under the assumption of normality, the estimators in Equations \eqref{eq:variance_estimator}-\eqref{eq:variance_estimator.Y} are also consistent.      The consistency of the estimators follows from classical results in estimation theory and it can be found in Proposition 5 in \cite{Wozny2024GaussianDS}.
 \end{remark}

\section{Empirical evidence}
\label{section:empirical}
The objective of this section is to explore the application of the 20/60/20 rule in financial contexts. The starting point of our analysis is that although the financial data are known to exhibit heavy-tailed behavior, we have observed that observations in the middle 60\% of the distribution often resemble a normal distribution. This observation motivates our focus on the middle subset, as it provides a region where established statistical tools for normally distributed data can be effectively applied. The 20/60/20 rule emphasizes treating each subset individually, enabling tailored approaches for distinct characteristics of the data. Using the properties of the middle 60\%, we can simplify certain analyses while still accounting for the heavier tails separately, thereby combining robustness with computational efficiency. Please note that all data in this section is downloaded from Yahoo Finance via \textbf{R} tidyquant package.

\subsection{Spatial clustering based on 20/60/20 rule in financial data} \label{section:one_dim}

In this section, we want to answer the question if spatial clustering based on the 20/60/20 rule provides any insight in the context of financial data.  As numerous studies suggest,  the distribution of logarithmic returns of market data tends to exhibit heavy-tailed bahaviour over various long time horizons, see, e.g., \cite{Lo1997}, \cite{Count2002}, \cite{Cheoljun2019}, and \cite{BIELAK2021102308} for details. This phenomenon is often attributed to factors such as global political events, human behavior, and significant correlations within the stock market. However, we now investigate whether the central portion of the observations can be adequately modeled by a normal distribution. This stylized fact is linked to the unconditional setup, and we aim to assess its validity in the spatial-based framework. To do so, we examine the Q–Q (quantile-quantile) plots. 

First, let us analyse the empirical distributions of returns of several major stock market indices. Figure~\ref{plot:QQ_index} presents empirical Q-Q  plots with normal quantiles for the daily returns of the S\&P 500 (SPX), Dow Jones Industrial Average (DJI), Nasdaq Composite (IXIC), FTSE 100 (FTSE), DAX (GDAXI), and Nikkei 225 (N225) over the 2014--2024 period. Given that these indices aggregate a large number of stocks, some form of normality in the central portion of the distribution is expected due to the Central Limit Theorem (CLT). The Q-Q plots confirm this, showing that observations in the middle part of the sample closely align with normality. Furthermore, the quantiles at levels $p=0.198$ and $q=1-p$ frequently coincide with the transition point where the empirical distribution shifts from normal to heavy-tailed behavior. This suggests that these quantiles may serve as markers for regime changes between different distributional properties.

\begin{figure}[htp!]
\begin{center}
\includegraphics[width=0.32\textwidth]{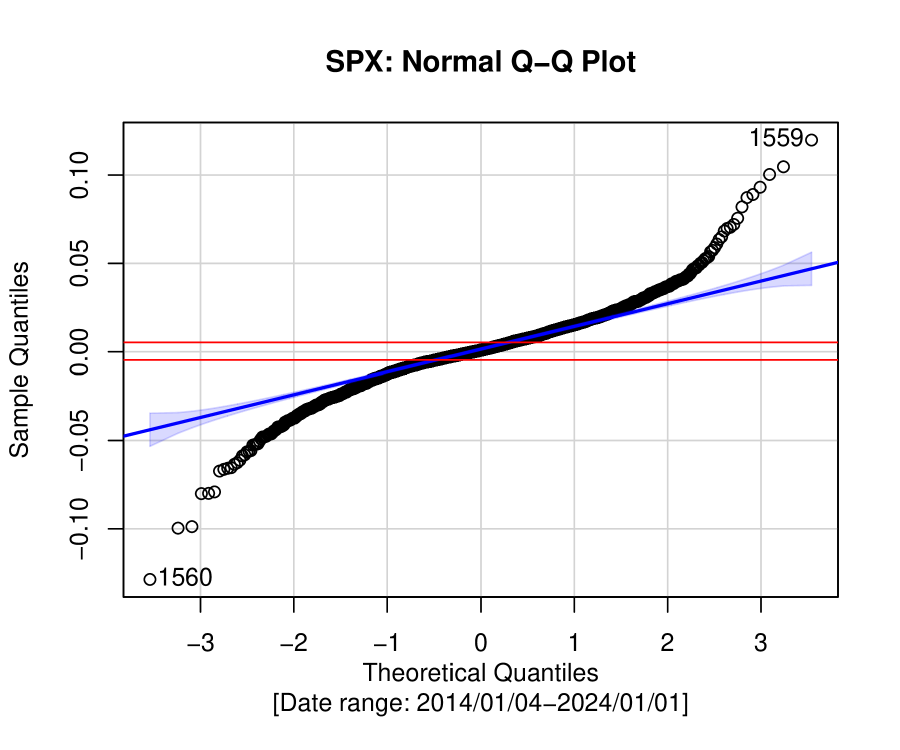}
\includegraphics[width=0.32\textwidth]{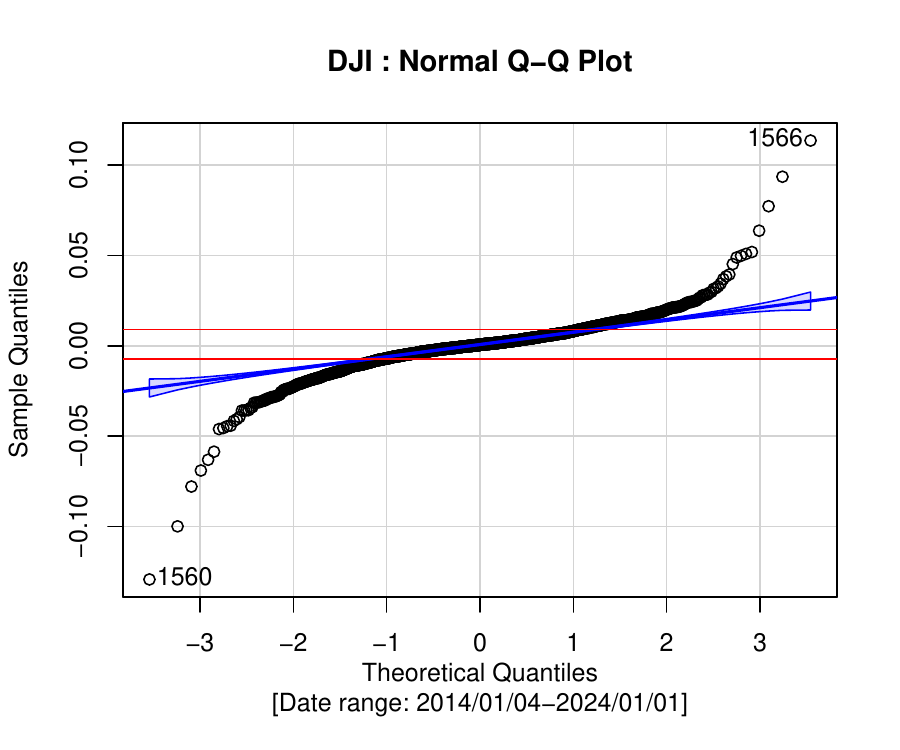}
\includegraphics[width=0.32\textwidth]{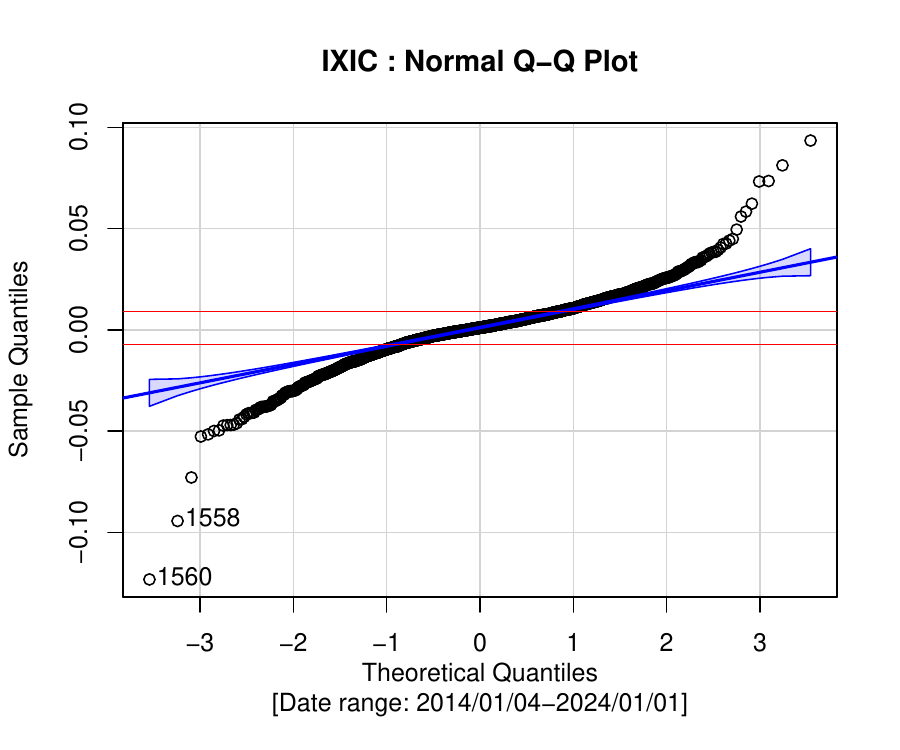}
\includegraphics[width=0.32\textwidth]{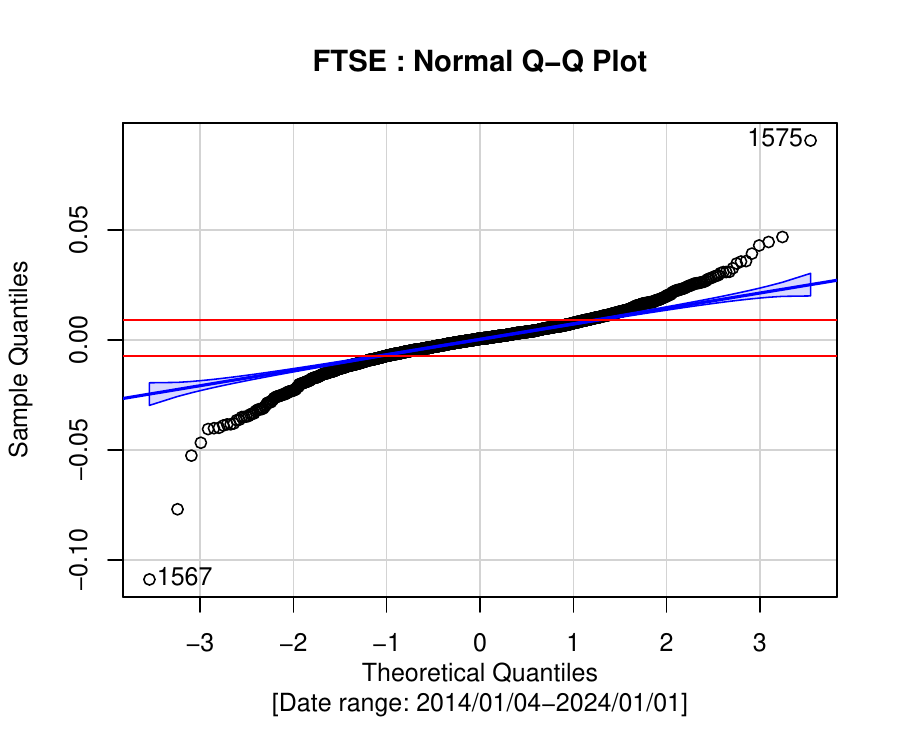}
\includegraphics[width=0.32\textwidth]{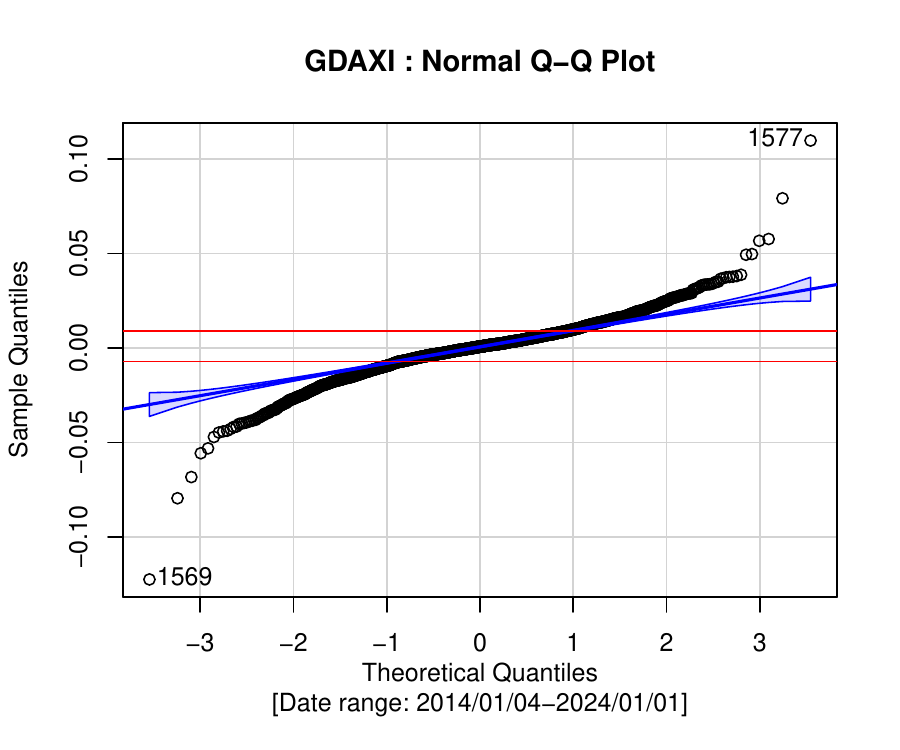}
\includegraphics[width=0.32\textwidth]{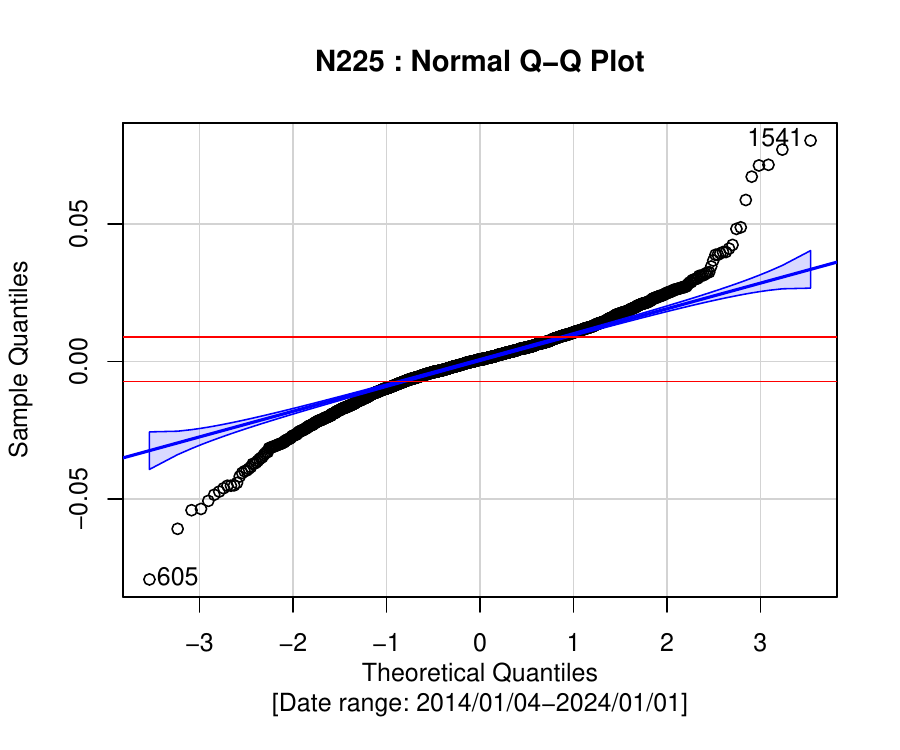}
\end{center}
\caption{The figures present Q–Q plots with normal theoretical quantiles of daily returns for S\&P 500 (SPX), Dow Jones Industrial Average (DJI), Nasdaq Composite (IZIC), FTSE 100 (FTSE), DAX (GDAXI) and Nikkei 225 (N225)  during the 2014--2024 period. The confidence intervals for the normal distribution are highlighted in blue, and the red lines represent the quantile levels $p=0.198$ and $q=1-p$. It can be observed that the area defined by these lines aligns with the normal distribution.}
\label{plot:QQ_index}
\end{figure}

Second, we decided to check if this phenomenon also occurs for the weekly and monthly rates of returns, as illustrated in Figure~\ref{plot:QQ_index_SPX}. However, due to the aggregation effects, the phenomenon is noticeably weaker, but still can be seen; we refer to Stylized fact nr 4. in \cite{Count2002} for details.

\begin{figure}[htp!]
\begin{center}
\includegraphics[width=0.32\textwidth]{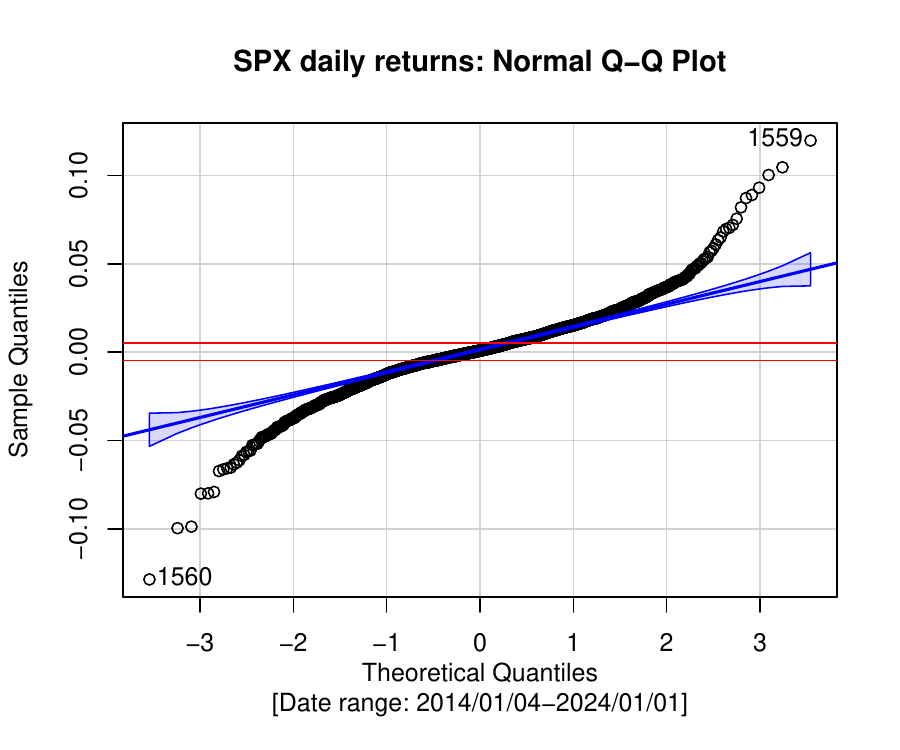}
\includegraphics[width=0.32\textwidth]{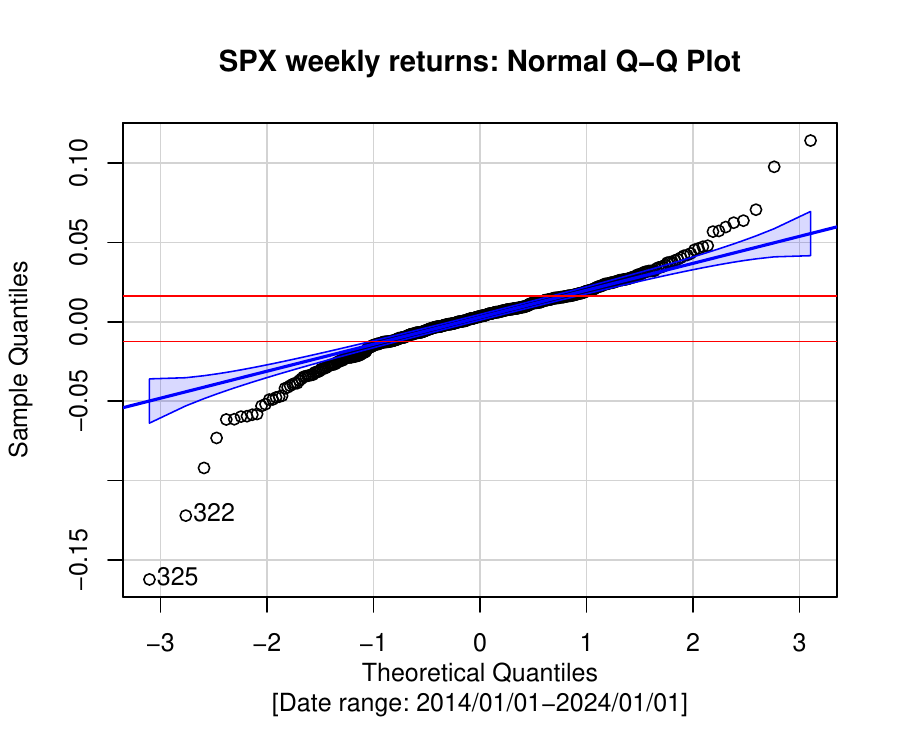}
\includegraphics[width=0.32\textwidth]{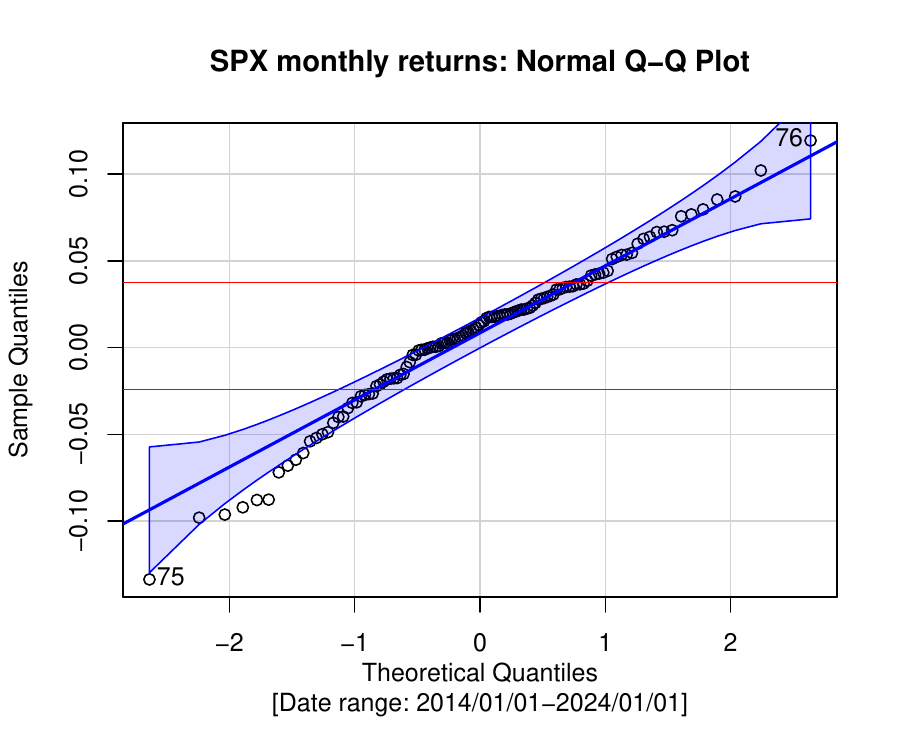}
\end{center}
\caption{The figures present Q–Q plots with normal theoretical quantiles of daily, weekly, and monthly returns for S\&P 500 (SPX)  during the 2014--2024 period. The confidence intervals for the normal distribution are highlighted in blue, and the red lines represent the quantile levels $p=0.198$ and $q=1-p$. It can be observed that the area defined by these lines aligns with the normal distribution.}
\label{plot:QQ_index_SPX}
\end{figure}
Finally, we verified if this phenomena is also visible on individual stock level. Figure~\ref{plot:QQ_stocks} shows examples of the Q–Q plot of the Apple (AAPL) and Amazon (AMZN) stocks in a $10$, $5$ and $1$ year period time daily and weekly returns. We can observe that, similarly to indices, the return rates of individual companies also exhibit a distribution close to normal within the 0.2 and 0.8 quantiles. As we can see, the points where there is an alleged switch between normal and heavy-tailed behavior are close to the 20/60/20 split and data tend to fit well to normal in the central region but diverge from normal in the tails.

\begin{figure}[htp!]
\begin{center}
\includegraphics[width=0.32\textwidth]{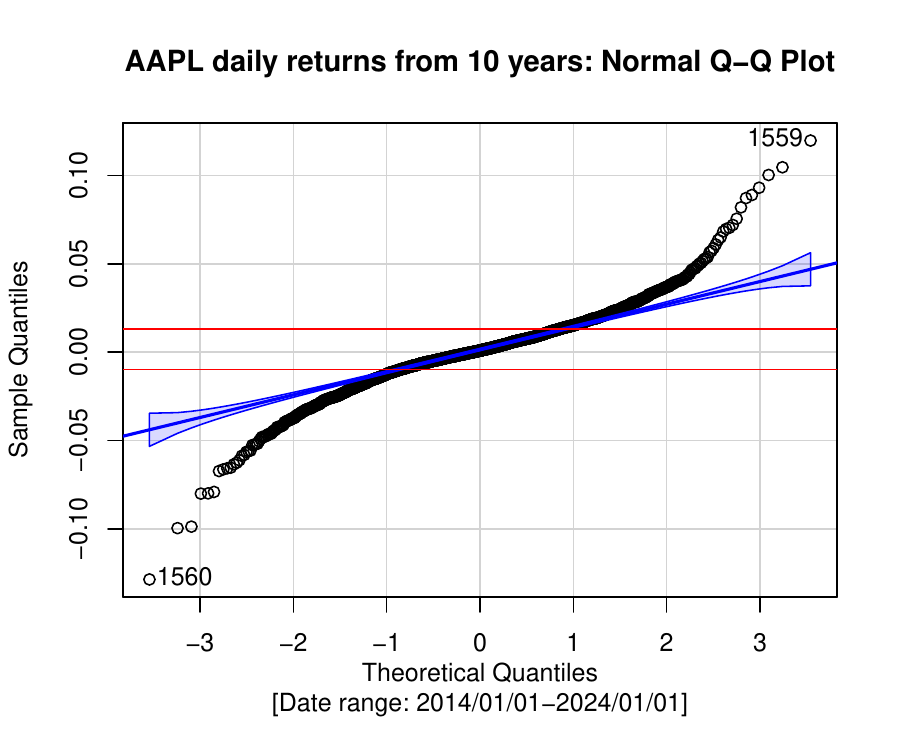}
\includegraphics[width=0.32\textwidth]{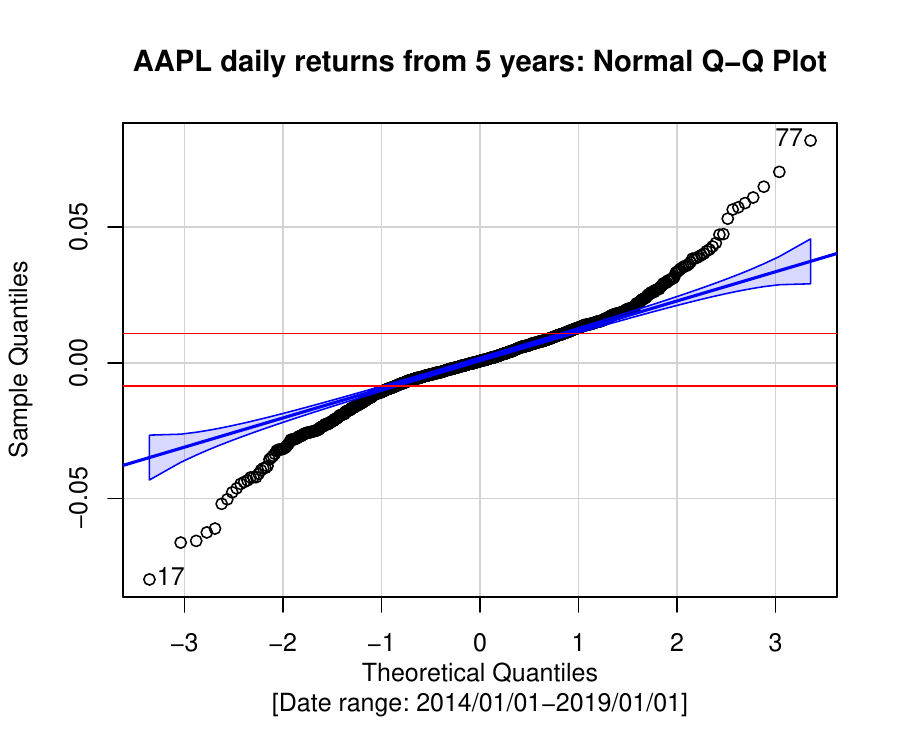}
\includegraphics[width=0.32\textwidth]{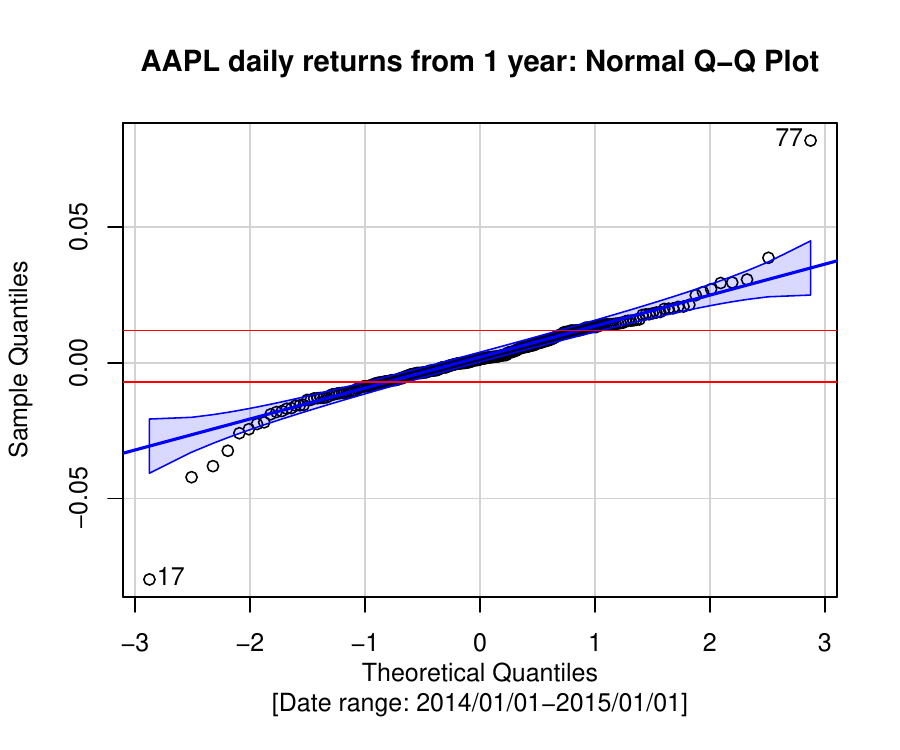}
\includegraphics[width=0.32\textwidth]{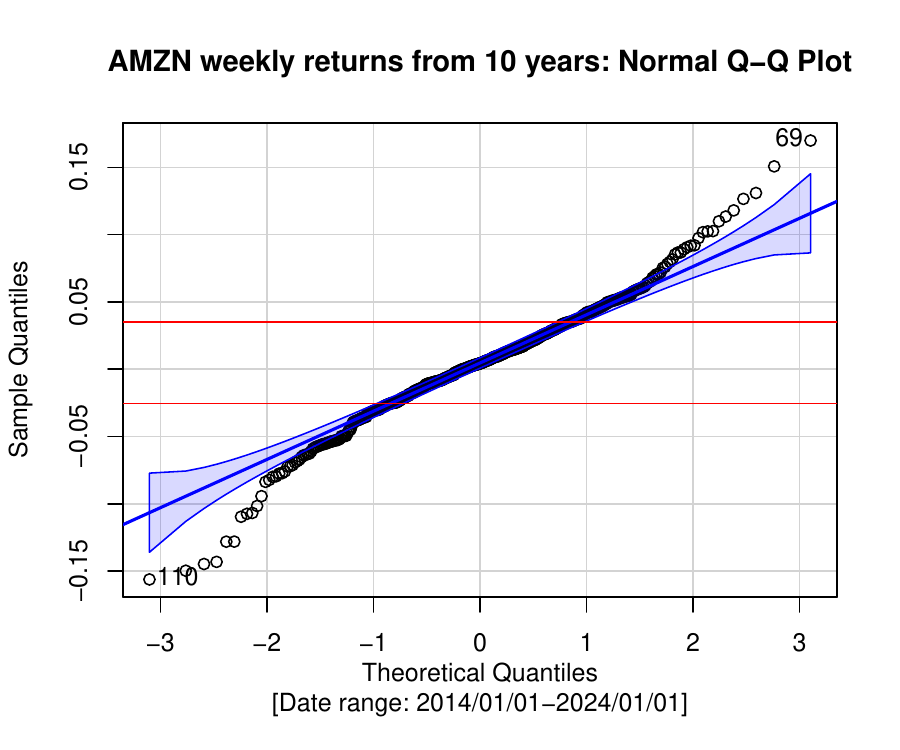}
\includegraphics[width=0.32\textwidth]{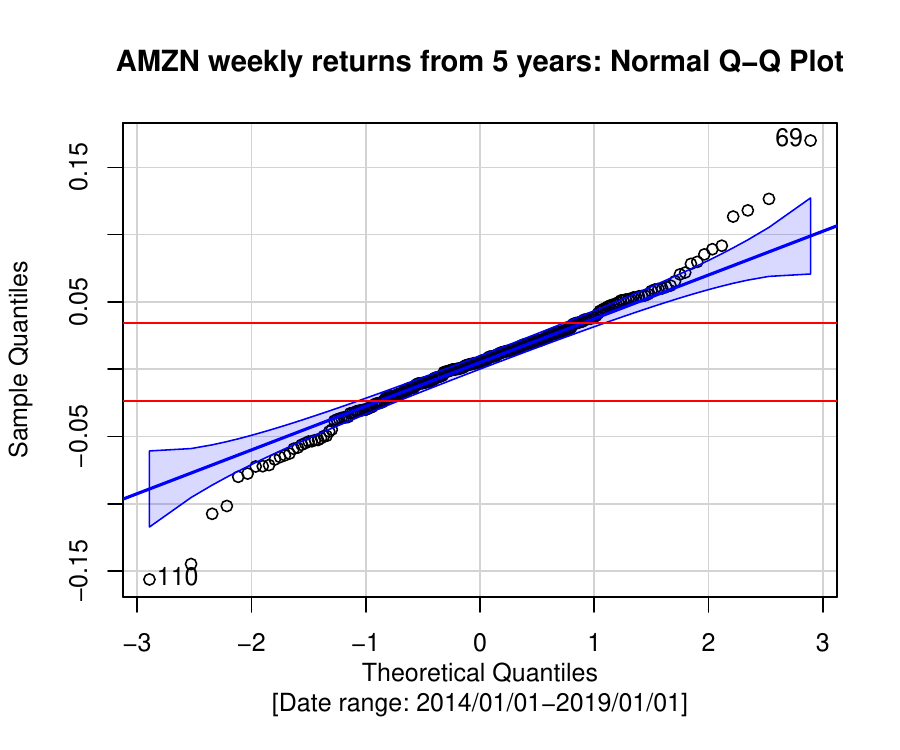}
\includegraphics[width=0.32\textwidth]{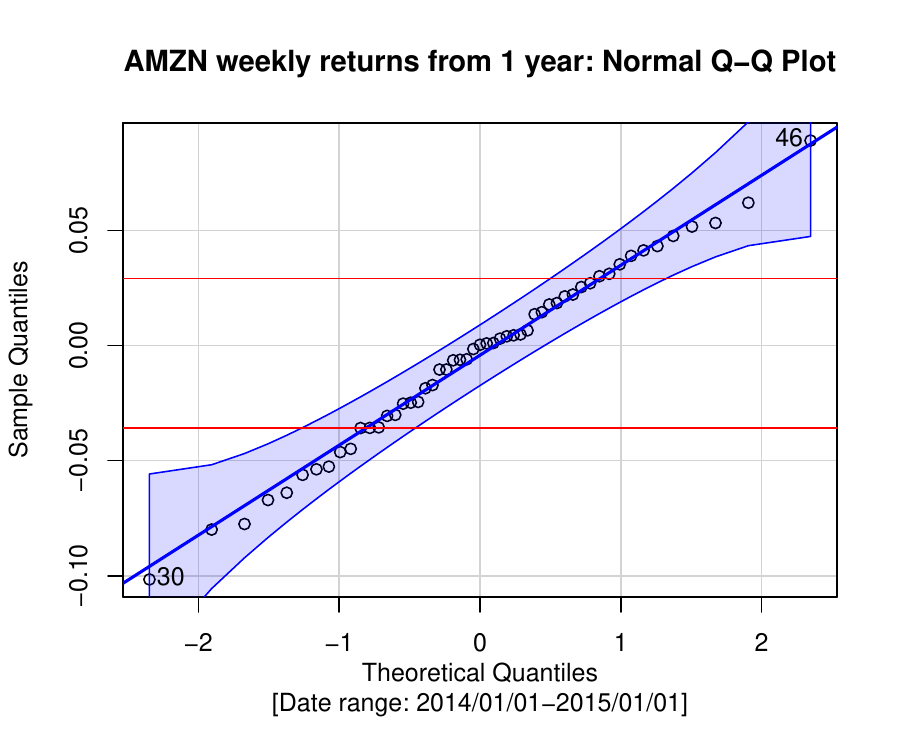}
\end{center}
\caption{The figures present Q–Q plots with normal theoretical quantiles of returns (daily and weekly) for Apple and Amazon during the 2014--2024, 2014--2019 and 2014--2015 periods. The confidence intervals for the normal distribution are highlighted in blue, and the red lines represent the quantile levels $p=0.198$ and $q=1-p$. It can be observed that the area defined by these lines aligns with the normal distribution.}
\label{plot:QQ_stocks}
\end{figure}
The 20/60/20 rule highlights a shift in spatial regimes within the distribution, indicating differing behaviors across conditional subsets that align with the rule's framework. While the central region can often be approximated by a normal distribution, the tail regions exhibit heavy-tailed dynamics, making them inherently difficult to manage. This observation is consistent with numerous financial stylized facts, including heavy-tailed distributions and the occurrence of rare, high-impact events, such as those described by the Black Swan phenomenon; see \cite{Lo1997, Cheoljun2019, BlackSwan}. To ensure the robustness of these findings, we have also examined additional stocks, indices, and various time horizons. The results consistently confirm that the conditional structure imposed by the 20/60/20 rule is a universal pattern, visible across a wide range of datasets. This suggests that the observed regime shifts are an inherent feature of financial return distributions rather than an artifact of a specific sample or market condition.


\subsection{Characterization of heavy-tailed behavior based on 20/60/20 rule}

In this section, our goal is to project the variance of the central region onto the unconditional variance to assess deviations from a balanced state. This approach allows us to define a distribution characterization metric that quantifies the heaviness of the tails relative to those induced by a normal distribution. By doing so, we provide a numerical measure of the effects discussed in Section~\ref{section:one_dim}, offering an alternative tail metric that does not rely on asymptotic tail behavior, as it is usually done. This enables a more flexible and data-driven assessment of distributional properties across different regimes. To do this, we introduce the statistic $S_n$ given, for any interval $(p,q)$, by
\begin{equation}
\label{s_stat}
    S_n := \sqrt{n} \left(\frac{\hat{\sigma}^2(\mathbf{Y})}{\bar{\sigma}^2(\mathbf{Y}) } -1 \right),
\end{equation}
where $\hat{\sigma}^2(\mathbf{Y})$ and  $\bar{\sigma}^2(\mathbf{Y}) $ are given by Equations \eqref{eq:standard.variance.Y.est} and \eqref{eq:variance_estimator.Y}.
In a nutshell, statistic $S_n$ measures the heaviness of the distribution tails when compared with the normal distribution. As we show now, this statistic is asymptotically normal as the sample size grows to infinity.

\begin{proposition}\label{pr:benchmark}

 Let $X\sim \mathcal{N}_d(\mu,\Sigma)$, and let $\mathbf{X}$ be a sample from $X$. Let us fix $p,q \in (0,1)$ such as $p < q$ and let $\mathbf{Y}$ and $S_n$ be given by \eqref{eq:sample_Y} and \eqref{s_stat}, respectively. Then, $S_n\overset{d}{\to }\mathcal{N}(0,\tau^2)$, as $n \to \infty$, where $\tau >0$ is a constant independent of $\mu$ and $\Sigma$.

\end{proposition}
\begin{proof}
The proof involves utilizing the multivariate central limit theorem and the delta method. Let
\begin{align}\label{eq.CLT}
P_n:=\sqrt{n}
\begin{bmatrix}
\hat{\sigma}^2(\mathbf{Y}) - \Var(Y) \\
\hat{\sigma}^2_B(\mathbf{Y}) - \Var(Y|Y\in B) \\
\end{bmatrix}.
\end{align}
We demonstrate that $P_n$ converges in distribution to a certain multivariate normal distribution as $n \to \infty$. We begin by presenting alternative expressions for the components of $P_n$.

First, setting $A_k := (Y_k - \mathbb{E}(Y))^2 - \Var(Y) $ and $p_n= -\sqrt n (\hat \mu(\mathbf Y) - \mathbb{E}(Y))^2$, we obtain
\begin{align}
        \sqrt{n}(\hat \sigma^2(\mathbf{Y}) - \Var(Y) )&= \sqrt{n}\frac{1}{n}\sum_{k=1}^n(Y_k-\mathbb{E}(Y))^2 - \sqrt{n}\Var(Y) -\sqrt{n}(\hat \mu(\mathbf Y) - \mathbb{E}(Y))^2.\nonumber\\
        & = \sqrt{n}\frac{1}{n}\sum_{k=1}^nA_k - \sqrt n (\hat \mu(\mathbf Y) - \mathbb{E}(Y))^2\nonumber \\
&=\sqrt n\frac{1}{n}\sum_{k=1}^nA_k + p_n \label{eq:Ak}.
        \end{align}
Next, note that $\sqrt{n}\left( \frac{1}{n} \sum_{k=1}^n A_k \right) \overset{d}{\to} \mathcal{N}(0,\eta^2)$, as $n\to \infty$ for some $\eta^2 > 0$ as a average of i.i.d random variables and by the central limiti theorem we have $p_n = o_p(1)$.

Now, combining Lemma 3 and Lemma 4 from \cite{JelPit2021} we get
\begin{align}
\sqrt{n}(\hat{\sigma}^2_B(\mathbf{Y}) - \Var(Y|Y\in B)) = \sqrt{n}\frac{n}{[nq] - [np]} \frac{1}{n}\sum_{k =1}^{n} B_k + t_n, \label{eq:JelPitB_i}
\end{align}
where $t_n \overset{\mathbb{P}}{\to} 0$ as $n \to \infty$ and, for any $k\in \mathbb{N}\setminus \{0\}$, we have
\begin{align*}
B_k &:= \left(\left(Y_k - \wo(Y|Y\in B)\right)^2-\Var(Y | Y\in B)\right)\mathbbm{1}_{\{\Phi^{-1}(p) < Y_k < \Phi^{-1}(q)\}} \\
&\phantom{=}+ \left(\mathbbm{1}_{\{Y_k \leq \Phi^{-1}(p)\}} - p\right)\left( \left(\Phi^{-1}(p) - \wo(Y|Y\in B)\right) - \Var(Y|Y\in B)\right) \\
&\phantom{=}+ \left(q-\mathbbm{1}_{\{Y_k \leq \Phi^{-1}(q) \}}\right)\left(\left(\Phi^{-1}(q)-\wo(Y|Y\in B)\right)^2 - \Var(Y|Y\in B)\right).
\end{align*}
Now, combining \eqref{eq:Ak} and \eqref{eq:JelPitB_i} we obtain
\begin{align}
P_n &=D_n
\sqrt{n} \frac{1}{n} \sum_{k = 1}^n
C_k
+  g_n, \quad n\in \mathbb{N},
\end{align}
where, for any $n\in \mathbb{N}$, we set $C_n:=[A_n, B_n]^T$, $g_n:=[p_n,t_n]^T$, and $D_n$  is $2\times 2$ diagonal matrix with the main diagonal given by $[1,\frac{n}{[nq] - [np]}]$. In particular, we get that $(C_n)$ is a family of i.i.d. square-integrable random vectors with zero expectation and $g_n\overset{\mathbb{P}}{\to} 0$ as $n\to\infty$. Thus, noting that $D_n$ converges to the diagonal matrix $ D$ with the main diagonal $[1,\frac{1}{q-p}]$ and using the multivariate central limit theorem combined with Slutsky's theorem, we get that
$P_n \overset{d}{\to}\mathcal{N}(\begin{bmatrix} 0, 0\end{bmatrix}^T, \Sigma)$, as $n\to\infty$, where 
\begin{align*}
\Sigma := D
\begin{bmatrix}
D^2(A_1) & \cov(A_1, B_1) \\
\cov(B_1, A_1) & D^2(B_1) \\
\end{bmatrix} D.
\end{align*}
Next, let us define $H_n:=[\hat\sigma^2(\mathbf{Y}),\hat{\sigma}^2_B(\mathbf{Y})]^T$, $H:=[\Var(Y),\Var(Y|Y\in B)]^T$, and note that $P_n=\sqrt{n}(H_n-H)\overset{d}{\to}\mathcal{N}([0,0]^T,\Sigma)$. Now, using Slusky's theorem for $g(x,y) := s(p,q) \frac{x}{y}$ and noting that
\begin{equation}
    \frac{s(p,q)\Var(Y)}{\Var(Y|Y \in B)} = \frac{\Var(Y|Y \in B)}{\Var(Y|Y \in B)} = 1,
\end{equation}
we obtain
\begin{align*}
    \sqrt{n} \left( g(H_n) - g(H) \right) =
    \sqrt{n} \left(\frac{s(p,q)\hat\sigma^2(\mathbf{Y})}{\hat{\sigma}^2_B(\mathbf{Y})} -\frac{s(p,q)\Var(Y)}{\Var(Y|Y\in B)} \right)=S_n\overset{d}{\to }\mathcal{N}(0,\tau^2), \quad n\to\infty,
\end{align*}
which concludes the proof.
\end{proof}

 In particular, noting that $S_n=\sqrt{n} \left(\frac{s(p,q)\hat{\sigma}^2(\mathbf{Y})}{\hat{\sigma}_B^2(\mathbf{Y}) } -1 \right)$, we get that $\left(\frac{s(p,q)\hat{\sigma}^2(\mathbf{Y})}{\hat{\sigma}_B^2(\mathbf{Y}) } -1 \right) \overset{\mathbb{P}}{\to } 0$ if the underlying sample is multivariate normal; this follows directly from Proposition~\ref{pr:benchmark}.

As previously discussed, the statistic $S_n$ serves as a general measure of tail heaviness that does not rely on tail asymptotics. Our primary objective here is to evaluate the effectiveness of $S_n$ as a tail measurement tool by applying it to financial datasets. Specifically, we analyze its values across various companies and indices for daily and weekly returns. Table~\ref{table:Stat} provides values of $S_n$ in comparison to the mean and the standard error of a normal sample, and the t-Student distributed samples with various choices of degrees of freedom computed for the corresponding sample lengths using $10000$ Monte Carlo simulations. By examining these results, we assess whether $S_n$ consistently captures the presence of heavy tails in financial returns, reinforcing its validity as an alternative tail metric that does not depend on extreme value theory or asymptotic behavior.

The analysis reveals that daily returns exhibit the most significant deviations from normality, closely resembling the characteristics of a t-Student distribution with low degrees of freedom, whereas the weekly (aggregated) data shows lower values, likely due to the influence of the central limit theorem. These findings align with the patterns observed in the Q-Q plots discussed in Section~\ref{section:one_dim}. Notably, while extended timeframes generally correspond to higher values across most cases, the 2019--2021 period (marked by the COVID-19 pandemic) often yields higher values compared to the broader 2019-2024 interval. This anomaly may be attributed to the heightened market volatility and disruptions caused by the pandemic crisis during those years.

\begin{table}[htp!]
\scalebox{0.75}{
\begin{tabular}{c|cccc|cccc}
\multicolumn{1}{l|}{}              & \multicolumn{4}{c|}{Daily returns}                                                                                                                                                                                                                                              & \multicolumn{4}{c}{Weekly returns}                                                                                                                                                                                                                                              \\
                                   & \begin{tabular}[c]{@{}c@{}}2014/01/01\\ 2024/01/01\end{tabular} & \begin{tabular}[c]{@{}c@{}}2014/01/01\\ 2019/01/01\end{tabular} & \begin{tabular}[c]{@{}c@{}}2019/01/01\\ 2024/01/01\end{tabular} & \begin{tabular}[c]{@{}c@{}}2019/01/01\\ 2021/01/01\\ (covid)\end{tabular} & \begin{tabular}[c]{@{}c@{}}2014/01/01\\ 2024/01/01\end{tabular} & \begin{tabular}[c]{@{}c@{}}2014/01/01\\ 2019/01/01\end{tabular} & \begin{tabular}[c]{@{}c@{}}2019/01/01\\ 2024/01/01\end{tabular} & \begin{tabular}[c]{@{}c@{}}2019/01/01\\ 2021/01/01\\ (covid)\end{tabular} \\ \hline
JNJ                                & 49.79                                                           & 29.94                                                           & 36.77                                                           & 40.10                                                                     & 8.15                                                            & 6.23                                                            & 3.44                                                            & 7.55                                                                      \\
PG                                 & 52.10                                                           & 27.39                                                           & 36.66                                                           & 42.28                                                                     & 12.18                                                           & 5.43                                                            & 10.68                                                           & 12.96                                                                     \\
V                                  & 52.05                                                           & 30.17                                                           & 37.75                                                           & 35.00                                                                     & 15.53                                                           & 6.45                                                            & 12.55                                                           & 12.00                                                                     \\
KO                                 & 57.69                                                           & 21.60                                                           & 46.78                                                           & 48.07                                                                     & 19.99                                                           & 2.59                                                            & 21.78                                                           & 30.45                                                                     \\
NVDA                               & 52.78                                                           & 51.60                                                           & 19.50                                                           & 18.39                                                                     & 9.73                                                            & 10.73                                                           & 2.68                                                            & 2.83                                                                      \\
PFE                                & 48.95                                                           & 29.15                                                           & 27.22                                                           & 33.87                                                                     & 12.58                                                           & 8.11                                                            & 4.85                                                            & 11.03                                                                     \\
CSCO                               & 57.75                                                           & 34.62                                                           & 41.75                                                           & 39.44                                                                     & 14.62                                                           & 10.69                                                           & 9.20                                                            & 5.31                                                                      \\
AMZN                               & 51.13                                                           & 44.40                                                           & 25.76                                                           & 20.68                                                                     & 11.00                                                           & 9.08                                                            & 6.73                                                            & 3.40                                                                      \\
AAPL                               & 46.97                                                           & 31.51                                                           & 27.55                                                           & 26.34                                                                     & 10.71                                                           & 8.48                                                            & 6.12                                                            & 9.24                                                                      \\ \hline
SPX                                & 80.55                                                           & 49.61                                                           & 45.75                                                           & 61.04                                                                     & 20.98                                                           & 12.05                                                           & 12.48                                                           & 19.12                                                                     \\
DJI                                & 87.57                                                           & 44.78                                                           & 64.07                                                           & 78.59                                                                     & 28.38                                                           & 13.46                                                           & 20.60                                                           & 27.88                                                                     \\
IXIC                               & 59.45                                                           & 35.36                                                           & 31.17                                                           & 40.33                                                                     & 9.78                                                            & 7.62                                                            & 3.18                                                            & 9.63                                                                      \\
FTSE                               & 54.07                                                           & 23.92                                                           & 49.99                                                           & 40.16                                                                     & 16.13                                                           & 5.72                                                            & 14.44                                                           & 17.81                                                                     \\
GDAXI                              & 54.07                                                           & 27.28                                                           & 48.95                                                           & 51.73                                                                     & 16.94                                                           & 3.75                                                            & 20.42                                                           & 16.57                                                                     \\
N225                               & 39.81                                                           & 39.46                                                           & 19.00                                                           & 22.97                                                                     & 11.56                                                           & 4.76                                                            & 13.73                                                           & 21.65                                                                     \\ \hline
normal                              & -0.03 (1.74)                                                            & 0.08 (1.74)                                                            & 0.08 (1.76)                                                            & -0.04 (1.76)                                                                      & -0.02 (1.75)                                                            & -0.24 (1.77)                                                            & 0.20 (1.77)                                                            & 0.09 (1.82)                                                                      \\
\multicolumn{1}{l|}{t-Student(3)}  & 66.26 (31.2)                                                           & 46.81 (40.04)                                                           & 47.97 (36.37)                                                           & 29.97 (35.2)                                                                      & 31.25 (42.39)                                                           & 21.43 (22.71)                                                           & 22.88 (23.82)                                                           & 13.23 (21.04)                                                                     \\
\multicolumn{1}{l|}{t-Student(5)}  & 22.04 (4.46)                                                            & 15.74 (4.20)                                                             & 15.67 (4.28)                                                            & 9.73 (4.30)                                                                       & 9.81 (4.06)                                                            & 7.24 (4.36)                                                            & 7.29 (4.21)                                                            & 4.49 (4.10)                                                                       \\
\multicolumn{1}{l|}{t-Student(10)} & 8.15 (2.32)                                                            & 5.86 (2.37)                                                            & 5.91 (2.36)                                                            & 3.59 (2.36)                                                                      & 3.60 (2.34)                                                            & 2.87 (2.44)                                                            & 2.78 (2.41)                                                            & 1.72 (2.45)                                                                      \\
\multicolumn{1}{l|}{t-Student(30)} & 2.31 (1.90)                                                             & 1.72 (1.89)                                                           & 1.70 (1.88)                                                            & 0.99 (2.07)                                                                      & 0.99 (1.89)                                                            & 0.95 (1.94)                                                            & 0.92 (1.93)                                                            & 0.57 (1.98)                                                                     
\end{tabular}
}
\caption{The table presents $S_n$ statistic values for a daily and weekly rate of returns for various stocks and the indices covering different periods from 01-01-2014 to 12-29-2023 and the empirical mean and standard deviation (in brackets) of $S_n$ for normal distribution and t-Student distribution with various degrees of freedom calculated by the $10000$ 
Monte Carlo simulation with corresponding lengths. The results for the normal sample are expected to be close to zero, consistent with theoretical predictions, while the results for stocks and indices suggest non-normality in the data. Notably, the daily stock data exhibits results similar to t-Student distribution with low degrees of freedom, whereas the weekly (aggregated) data shows lower values, likely due to the influence of the Central Limit Theorem.}
\label{table:Stat}
\end{table}
The $S_n$ statistic offers an easy-to-implement tool for characterizing extreme risks in financial datasets, contributing to more robust risk management strategies. In particular, it can be used to detect heavily-tailed distributions and identify outliers, providing a systematic approach to assessing the presence of extreme market events.

\subsection{Application of 20/60/20 rule to portfolio optimization} \label{s:Markowitz}

The traditional Markowitz portfolio optimization problem assumes that the variance-covariance matrix effectively quantifies risk; see \cite{Markowitz1952} for details. However, as shown in previous sections, while the central portion of the distribution is approximately normal, the presence of heavy tails raises concerns about the adequacy of variance-covariance as a risk measure, potentially leading to estimation errors. In this example, we test whether applying different methods to distinct subsets of the data (placing greater emphasis on the central spatial set) yields different optimization results. This allows us to assess the practical implications of the 20/60/20 framework beyond simple tail-to-center comparisons. To achieve this, we modify the traditional Markowitz approach using Theorem~\ref{th:OurTh}, adapting the optimization procedure to account for the distinct distributional properties across different regions of the return distribution.

The 20/60/20 rule suggests that focusing on the middle 60\% of a dataset can improve decision-making efficiency. In portfolio optimization, this implies that targeting middle-performing assets, rather than emphasizing only extremes, may lead to more balanced and effective outcomes. Our modified approach tests this principle, evaluating whether it provides a more stable and practical alternative to traditional mean-variance optimization.

In our simplified setup, we focus exclusively on spatial conditioning, applying controls derived from the middle 60\% of assets. Although this approach demonstrates the potential utility of the 20/60/20 rule in portfolio management, it does not account for tail behavior or market extremes, which can be critical in real-world scenarios. We chose not to include these refinements in this study to keep the focus on demonstrating the basic applicability of the 20/60/20 concept. Incorporating such enhancements would require the use of specific temporal filters, dynamic adjustments, and a more intricate framework, which, while valuable, goes beyond the scope of this initial investigation. A more sophisticated implementation would require adjusting the trading strategy to account for temporal and contextual market conditions, such as risk aversion tailored to periods of high volatility or other extreme events and this is left for future study.

\subsection*{Portoflio optimization model}

In this section, we use $X=(X_1,\ldots, X_d)$ to represent market returns of some $d$ assets with unknown distribution, and $Y$ to denote a random variable representing the benchmark (e.g. the returns of the market portfolio). Note that while, in contrast to the previous section, we do not assume a priori that $Y$ is a linear combination of the constituents of $X$, this is often effectively the case as benchmarks are calculated as a weighted average of returns of the stock market.
We consider a one-stage investment problem with a typical assumption of the Markowitz optimization problem (see Chapter~6 in ~\cite{elton2014modern} for details), where the objective is to determine the optimal portfolio weight vector $w=(w_1,\dots,w_d) \in S$, where the set of admissible allocations
\[
S:= \left\{ (w_1,\dots,w_d) \in \mathbb{R}^d: \sum_{k=1}^d w_k = 1 \right\}
\]
is defining the feasible set of weights; note that we allow short-selling. Given $w\in S$, the corresponding portfolio return could be expressed as $R=\langle w, X\rangle$. For brevity, we use $\Sigma$ and $\mu$ to denote the (unknown) covariance matrix and the expected value of the vector of the rates of return. 
Thus, the total return of the portfolio is a random variable $R:= \sum_{k =1}^{d} w_kX_k$ with mean $\mu^Tw$ and variance $\sigma^2_R:= w^T\Sigma w$.

The classical mean-variance optimization problem that finds the weights connected with the lowest variance of the portfolio with expected profits bigger than the expected (desired) rate of return $c > $ is given by
\begin{align} \label{eq:markowitz:objective}
    & \textrm{ minimize } \frac{1}{2} w^T\Sigma w \nonumber\\
    & \textrm{ subject to } \mu^Tw \geq c,  \textrm{ and } e^Tw = 1,
\end{align}
where $e$ is a d-dimensioned unit vector. If $\Sigma$ and $\mu$ are known, then the problem \eqref{eq:markowitz:objective} can be solved explicitly; see Section~2, Chapter~6 in \cite{elton2014modern} for details. In this case, the optimal solution is represented by a convex combination of two key portfolios, that is the minimum variance portfolio 
and the market portfolio, where
\begin{equation*}
    w_{\textrm{min-var}} = \frac{\Sigma e^{-1}}{e^T\Sigma^-1e},\,
    w_{mk} = \frac{\Sigma^{-1}\mu}{e^T\Sigma^{-1}\mu},  \textrm{ and }
    \alpha = \frac{c - \mu^Tw_{\textrm{min-var}}}{\mu^T(w_{mk} - w_{ \textrm{min-var}})}.
\end{equation*}
The optimal allocation is given by
\begin{equation} \label{eq:w_opt}
    w_{opt} = (1-\alpha)w_{ \textrm{min-var}} + \alpha w_{mk}.
\end{equation}

Since, in practice, we do not know the true distribution of $X$, we need to estimate both $\Sigma$ and $\mu$.
We use $\mathbf{X}$ to represent a specified subset of simple returns of stocks that contribute to our portfolio derived from historical market observations and $\mathbf{Y}$ to denote the sample for a benchmark stock index, aligning this framework with the setup introduced in Section~\ref{section:math}. For instance, $\mathbf{X}$ might correspond to five specific constituents of the S\&P 500 index, while $\mathbf{Y}$ represents the index itself. 
The most common solution to approximate the distribution of $X$ and consequently $R$ is the plug-in approach in which we first estimate $\widehat{\Sigma}$ and $\hat{\mu}$ and later plug those values into \eqref{eq:w_opt} to approximate the solution of \eqref{eq:markowitz:objective}. In the classical setup, unconditional sample estimators are used which results in the optimal allocation given by
 \begin{align}
    \hat{w}&:= (1-\hat{\alpha})\hat{w}_{min-var} + \hat{\alpha}\hat{w}_{mk}, \label{eq:w_opt_est} 
\end{align}
where
\begin{align}
    \hat{w}_{\textrm{min-var}} &:= \frac{\widehat{\Sigma} e^{-1}}{e^T\widehat{\Sigma}^{-1}e}, \textrm{ and }
    \hat{w}_{mk} := \frac{\widehat{\Sigma}^{-1}\hat{\mu}}{e^T\widehat{\Sigma}^{-1}\hat{\mu}}, \textrm{ and } 
    \hat{\alpha} := \frac{c - \hat{\mu}^T\hat{w}_{\textrm{min-var}}}{\hat{\mu}^T(\hat{w}_{mk} - \hat{w}_{\textrm{min-var}})}.
\end{align}
 In this paper, following Section~\ref{section:math}, we introduce alternative estimators based on conditional moments. More specifically, we set
\begin{align}
    \overline{w}&:= (1-\overline{\alpha})\overline{w}_{min-var} + \overline{\alpha}\overline{w}_{mk} \textrm{, where }\\
    \label{eq:w_opt_con_var}
     \overline{w}_{\textrm{min-var}} & := \frac{\overline{\Sigma} e^{-1}}{e^T\overline{\Sigma}^{-1}e},
    \overline{w}_{mk} := \frac{\overline{\Sigma}^{-1}\hat{\mu}_B}{e^T\overline{\Sigma}^{-1}\hat{\mu}_B}, 
    \overline{\alpha} := \frac{c - \hat{\mu}_B^T\overline{w}_{\textrm{min-var}}}{\hat{\mu}_B^T(\overline{w}_{mk} - \overline{w}_{\textrm{min-var}})},
\end{align}
where $B$ corresponds to the middle 60\% observation with respect to the benchmark, $\overline{\Sigma} = \overline{\Sigma}(X)$ is defined as the (unconditional) variance-covariance matrix estimator \eqref{eq:con_covariance_estimator} and $\hat{\mu}_B = \hat{\mu}_B (X)$ is defined in \eqref{eq:cond_exp_emp}. While $ \hat{w}$ is calculated classically with the standard estimators, $ \overline{w}$  calculated based on Equation \eqref{eq:w_opt_con_var} reflects the 20/60/20 rule. It is important to note that if the estimated minimal variance portfolio satisfies the required rate of return $c$, then this portfolio is considered optimal. Specifically, if either $\hat \mu ^T\hat w_{min-var} \geq c$ or $\hat \mu_B ^T\overline w_{min-var} \geq c$, then the respective portfolio weight $\hat w$ and $\overline w$ are equal to $\hat w_{min-var}$ and $\overline w_{min-var}$, respectively.

The approach presented in \eqref{eq:w_opt_con_var} can be viewed in several ways. First, conditioning the distribution of the stock market index can be seen as a way to account for market conditions within our portfolio. Additionally, this approach can be interpreted as an alternative estimator of the covariance matrix in a post-crisis scenario, where historical data contain many outliers that we believe should no longer influence our results. In such a scenario, accounting for excessive correlation could negatively impact the performance of our portfolio.
Since both estimators of the covariance matrix are consistent with a sample from a normal distribution, both estimators should converge to their theoretical values. For completeness, we present a simple proposition, which shows that both methods are assymptotically equivalent under normality assumption.

\begin{proposition}
Let $\mathbf{X}$ be a i.i.d. sample from $X \sim \mathcal{N}_d(\mu,\Sigma)$, $\hat \mu ^T\hat w_{min-var} < c$ and $\hat \mu_B ^T\overline w_{min-var} < c$, then for the quantiles $q = 0.2$, $p = 0.8$ and the corresponding Borel set B we have $\|\hat{w} -\overline{w}\| \overset{\mathbb{P}}{\to} 0$ as $n\to\infty$.
\end{proposition}

\begin{proof}
    First, note that by the Proposition 5 in \cite{Wozny2024GaussianDS}, the random variables $\overline{\Sigma}$, $\hat{\Sigma}$, $\hat{\mu}$, and $\hat{\mu}_B$ are consistent estimators of $\Sigma$, $\Sigma$, $\mu$, and $\mu_B$, respectively. Additionally, since $X$ is symmetric and $q = 1-p$, we have $\mu = \mu_B$.
    Because $\hat{w}$ and $\bar{w}$ are calculated as continuous functions of $\hat\Sigma$, $\bar\Sigma$,  $\hat\mu$, and $\hat\mu_B$, both $\overline{w}$ and $\hat{w}$ are also consistent. Thus, noting that
    \begin{equation}
    \|\overline{w} - \hat{w}\| = \|(\overline{w} -w) - (\hat{w} -w\|) \leq \|\overline{w} -w \| + \|\hat{w} -w\|
    \end{equation}
and recalling $\|\overline{w} -w \|\overset{\mathbb{P}}{\to} 0$ and $\|\hat{w} -w \| \overset{\mathbb{P}}{\to} 0$, we conclude the proof.
\end{proof}

\subsection*{Performance evaluation metric}

Our primary evaluation metric is the classical Sharpe Ratio (SR), which measures the risk-adjusted return of a portfolio. Given a sequence of portfolio allocations $\{w_t\}$ over a sample of length $T$, we obtain a corresponding sequence of realized portfolio returns $\{R_t\}_{t=1}^{T}$. The Sharpe Ratio is formally defined as a function of the return sequence $(R_t)$, given by
\begin{align}
\label{eq:SR}
SR(R_t) := \frac{\mu_R}{\sigma_R} \times \sqrt{250},
\end{align}
where $\mu_R$ represents the sample mean return over the period $T$
\begin{align}
\mu_R(R_t) = \frac{1}{T} \sum_{t=1}^T R_t,
\end{align}
and $\sigma_R$ is the standard deviation of returns, given by
\begin{align}
\sigma_R(R_t) = \sqrt{\frac{1}{T} \sum_{t=1}^T (R_t - \mu_R)^2}.
\end{align}
The factor 250 annualizes the number of trading days in a year. This metric provides a standardized measure of portfolio performance by balancing returns against associated risk, see \cite{sharpe1964capm} for further discussion. In our analysis, we compare the Sharpe Ratios of two different portfolio return sequences, that is $(\tilde{R}_t)$, which corresponds to the returns generated by the modified approach based on the 20/60/20 rule, and $(\bar{R}_T)$, which represents returns from the classical Markowitz optimization. Specifically, we compute and contrast $SR(\tilde R_t)$ and $SR(\bar R_t)$ to evaluate whether incorporating the central spatial-set in portfolio construction leads to improved risk-adjusted performance. 
\subsection*{Backtesting Design}
To compare the proposed method with the classical counterpart, we employ real financial market data over a various historical periods 01.01.2014--01.01.2024, 01.01.2019--01.01.2024 and 01.01.2014--01.01.2019. Our evaluation framework employs a rolling-window backtesting approach to estimate portfolio weights for both methodologies. Specifically, we utilize a learning period of $t=120$ trading days, during which the estimators of portfolio weights (denoted $\bar w$ for the proposed method and $\hat w$ for the classical approach) are computed. These weights are subsequently applied to the Markowitz portfolio optimization strategy. To address market non-stationarity, we maintain each set of portfolio weights for a fixed evaluation period of $l = 60$ trading days, after which they are re-estimated using the most recent 120-day rolling window. For each strategy, the realized portfolio returns are computed over the entire evaluation period.

To establish a dynamic benchmark for evaluating portfolio performance, we define the expected daily rate of return threshold $c_i$ at at each iteration $i$ as $c_i = \max(3 \mu_{Y,i},0.0005),$ where $\mu_{Y,i}$ represents the average daily return of a reference market index computed over the most recent backtesting sample at the iteration $i$. This threshold ensures that only strategies with a theoretically significant outperformance relative to the index are considered, while maintaining a slightly positive benchmark during economic downturns. Our analysis might be summed up in a simple algorithm presented in Figure~\ref{fig:backtest_algo}.

\begin{figure}[H]
    \centering
    \begin{algorithmic}[1]
        \State \textbf{Input:} Historical market data i.e. $\mathbf{X}$ and $\mathbf{Y}$, learning period $t$, investment period $l$ and $I$ number of learning periods
        \For{each iteration \( i \) to $I$} 
            \State Using \eqref{eq.standard.mean}  and \eqref{eq:standard.covariance.est}, estimate $\widehat \mu$, $\widehat \Sigma$ and $\mu_{Y,i}$ over \( t \) days
            \State Compute the desired expected return \( c_i = \max(3 \mu_{Y,i}, 0.0005) \)        
            \State Using \eqref{eq:cond_exp_emp} and \eqref{eq:cond_cov_emp_X}, estimate 
            $\widehat \mu_B$ and $\widehat{\Sigma}_B$ over $t$ days
            \State Use Theorem~\ref{th:OurTh} to reconstruct the unconditional covariance matrix $\overline \Sigma$ (Formula~\eqref{eq:variance_estimator})
            \State Using equations~\eqref{eq:w_opt_est} and \eqref{eq:w_opt_con_var}, compute portfolio weights \( \hat{w} \) and \( \overline{w} \)
            \State Invest according to the plug-in Markowitz framework for both estimators over \( l\) days
        \EndFor
    \end{algorithmic}
    \caption{Rolling-window backtesting procedure for portfolio optimization}
    \label{fig:backtest_algo}
\end{figure}

\subsection*{Market data analysis}

For both estimators, we pick six sets of randomly selected companies from the S\&P 500 index. The list of companies in each basket is presented in Table~\ref{table:BasketNames}. We consider three periods of investment: $01.01.2014-01.01.2024$, $01.01.2019-01.01.2024$, and $01.01.2014-01.01.2019$. In Table~\ref{tableBasket}, we present the Sharp Ratio for each basket and each period. The symbol "M" indicates the classical Markowitz estimator based on the unconditional sample covariance estimator and "CM" indicates the Markowitz estimator based on the rescaled conditional variance matrix. Results with the best performance are bolded.

\begin{table}[htp!]
\begin{tabular}{c|cccccccccc}
              & 1    & 2    & 3    & 4    & 5    & 6     & 7   & 8    & 9   & 10   \\ \hline
Basket 1 (B1) & JNJ  & PG   & VZ   & KO   & DIS  & PFE   &     &      &     &      \\
Basket 2 (B2) & MSFT & AAPL & NVDA & AMZN & META & GOOGL & HCA & TSLA & JPM & NFLX \\
Basket 3 (B3) & ZTS  & SYK  & PLD  & CMR  & GS   &       &     &      &     &      \\
Basket 4 (B4) & BKNG & MAR  & NKE  & DIS  & MDT  & XOM   & PEP &      &     &      \\
Basket 5 (B5) & RTX  & INTU & DECK & USB  & GILD & LMT   &     &      &     &      \\
Basket 6 (B6) & TMO  & LMT  & TT   & CAT  &      &       &     &      &     &     
\end{tabular}
\caption{The table presents the companies included in particular baskets.}
\label{table:BasketNames}
\end{table}

\begin{table}[htp!]
\begin{tabular}{c|cc|cc|cc|}
   & \multicolumn{2}{c|}{01.01.2014-01.01.2024}                           & \multicolumn{2}{c|}{01.01.2019-01.01.2024}                           & \multicolumn{2}{c|}{01.01.2014-01.01.2019}                           \\ \cline{2-7} 
   & \multicolumn{1}{c}{M} & \multicolumn{1}{c|}{CM} & \multicolumn{1}{c}{M} & \multicolumn{1}{c|}{CM} & \multicolumn{1}{c}{M} & \multicolumn{1}{c|}{CM} \\ \hline
B1 & 0.41                           & \textbf{0.76}                       & 0.27                           & \textbf{0.53}                       & 0.57                           & \textbf{1.02}                       \\ \hline
B2 & 1.03                           & \textbf{1.14}                       & 0.91                           & \textbf{1.15}                       & \textbf{1.27}                  & 1.26                                \\ \hline
B3 & 0.39                           & \textbf{0.61}                       & 0.33                           & \textbf{0.41}                       & 0.34                           & \textbf{0.94}                       \\ \hline
B4 & \textbf{0.73}                  & \textbf{0.73}                       & \textbf{0.54}                  & 0.52                                & 1.01                           & \textbf{1.12}                       \\ \hline
B5 & 0.42                           & \textbf{0.75}                       & 0.45                           & \textbf{0.73}                       & \textbf{0.64}                  & 0.57                                \\ \hline
B6 & \textbf{0.76}                  & 0.62                                & 0.44                           & \textbf{0.86}                       & \textbf{1.12}                  & 0.94                                \\ \hline
\end{tabular}
\caption{The table presents the Sharpe Ratio results for portfolios of six stock baskets picked randomly over different investment horizons. The strategy based on the assumption of normality in the central set generally
yields significantly better results (bolded) than its classical counterpart in most of the cases considered.}
\label{tableBasket}
\end{table}

The results, presented in Table~\ref{tableBasket}, demonstrate that the strategy based on the reconstructed covariance matrix generally outperforms its classical counterpart in most of the cases considered. Specifically, the proposed method achieved superior results in 12 out of 18 cases, with performance nearly doubling in two instances. Notably, even in cases where the classical method (M) outperformed the proposed approach (CM), the differences were relatively small. 
These findings support the applicability of the 20/60/20 rule in finance, which may offer advantages over using variance estimated from unconditional data. The latter approach can introduce issues such as risk spillovers, distorted correlations, and reduced control over the optimization process. In particular, the optimal allocation obtained using the CM model delivered robust performance during the 2019-–2024 period—a time characterized by major crises, including the COVID-19 pandemic and Russia's invasion of Ukraine. This suggests that conditioning on the benchmark does not disrupt key market relationships but rather preserves and accurately reflects the underlying dynamics, which tend to become more pronounced during periods of financial turbulence.
These findings support the applicability of the 20/60/20 rule in finance, which may offer advantages over using variance estimated from unconditional data. The latter approach can introduce issues such as risk spillovers, distorted correlations, and reduced control over the optimization process. In particular, the optimal allocation obtained using the CM model delivered robust performance during the 2019-–2024 period—a time characterized by major crises, including the COVID-19 pandemic and Russia's invasion of Ukraine. This suggests that conditioning on the benchmark does not disrupt key market relationships but rather preserves and accurately reflects the underlying dynamics, which tend to become more pronounced during periods of financial turbulence.

\begin{figure}[htp!]
\begin{center}
\includegraphics[width=0.32\textwidth]{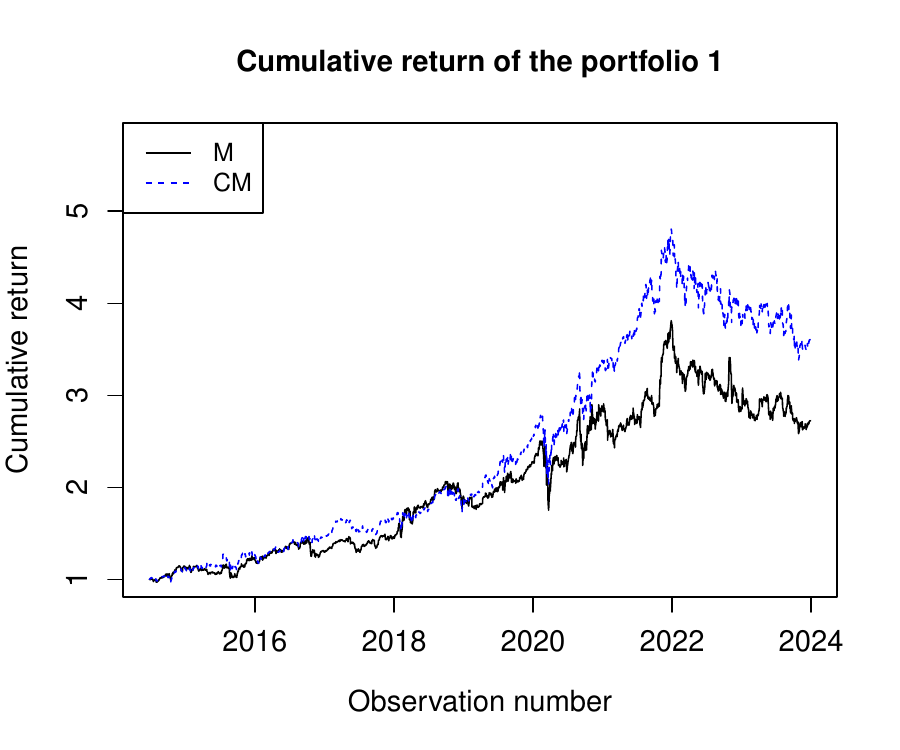}
\includegraphics[width=0.32\textwidth]{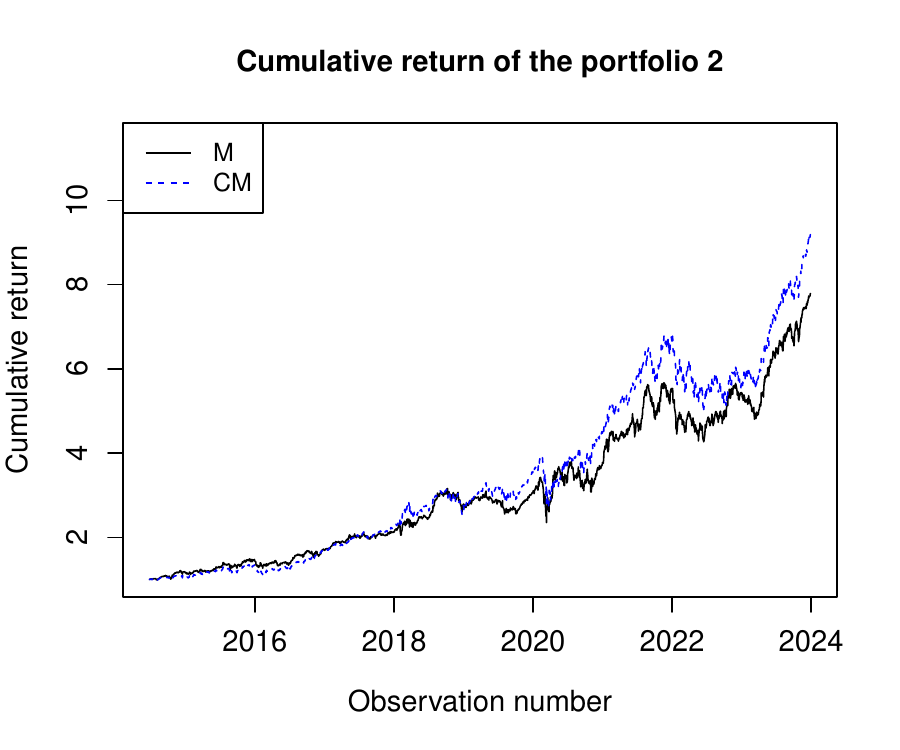}
\includegraphics[width=0.32\textwidth]{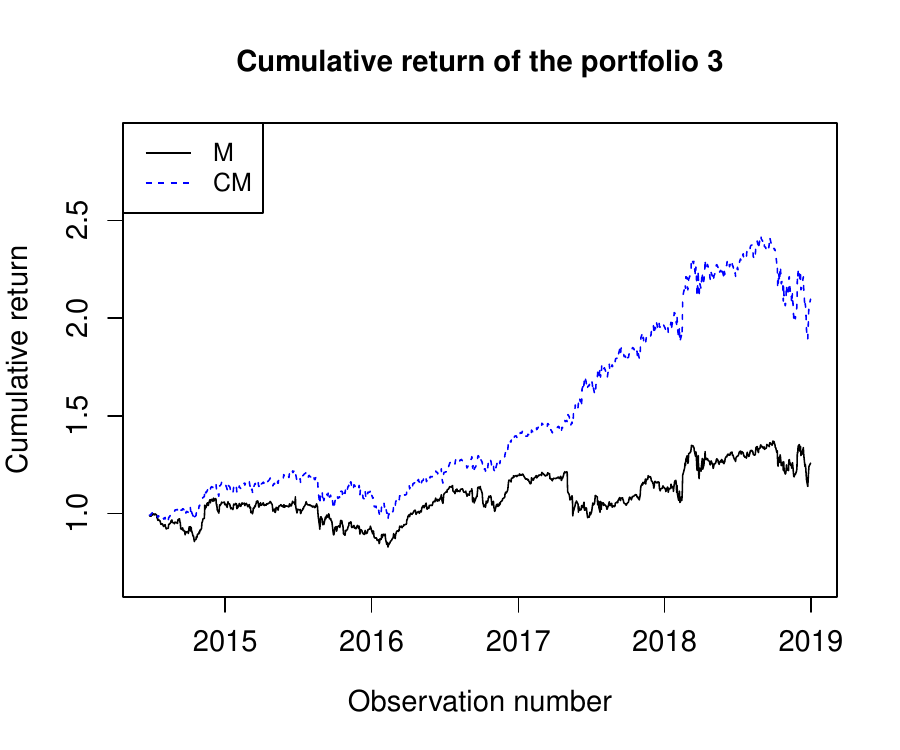}
\includegraphics[width=0.32\textwidth]{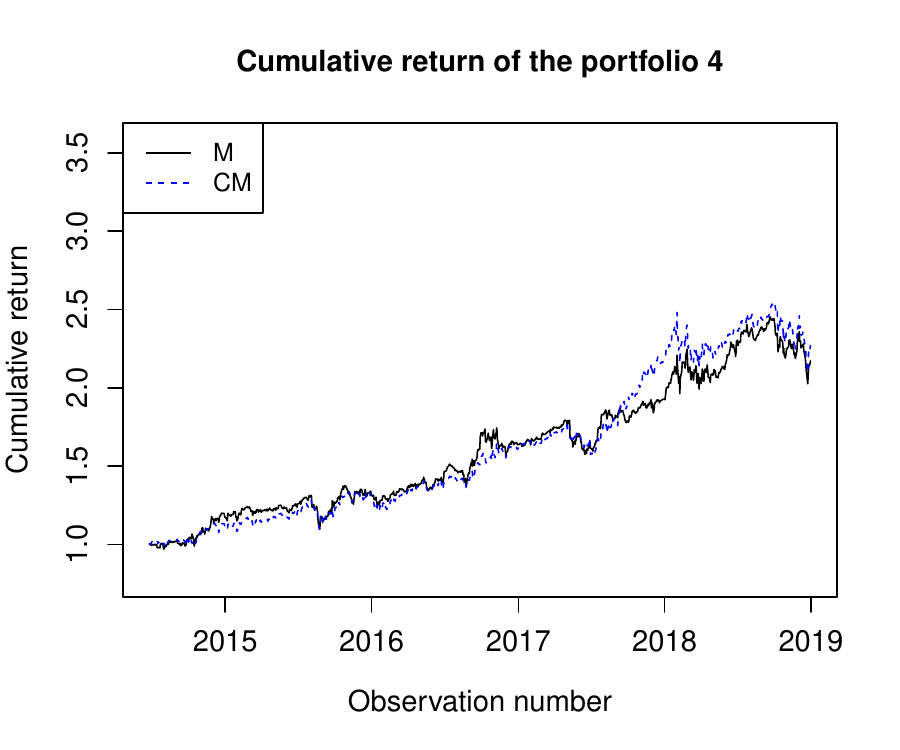}
\includegraphics[width=0.32\textwidth]{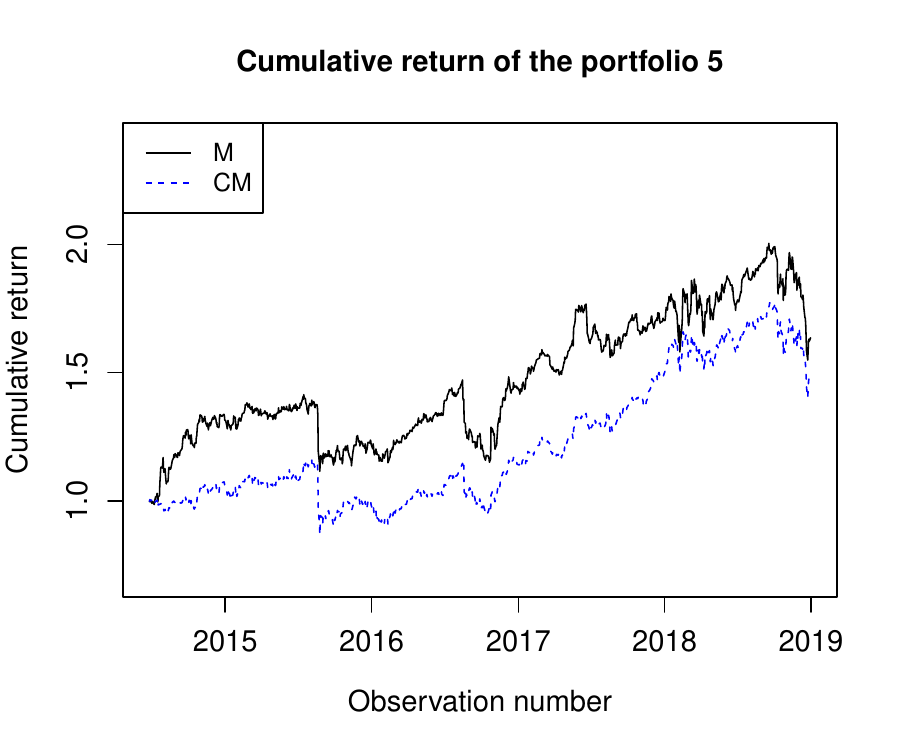}
\includegraphics[width=0.32\textwidth]{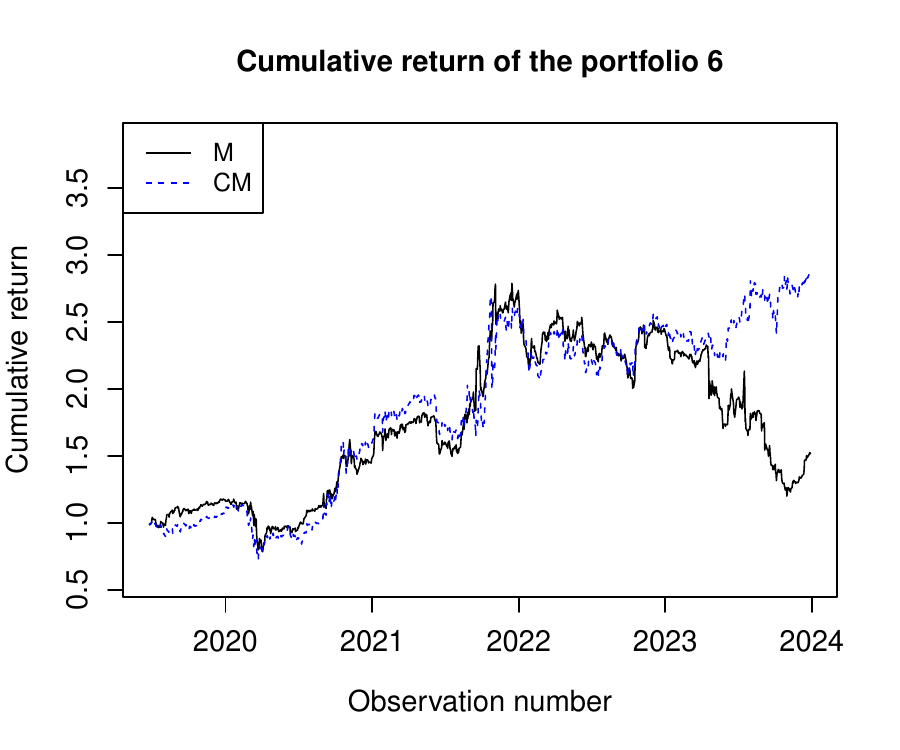}

\end{center}
\caption{The figure shows the cumulative return over the validation period $t=120$ for the M strategy (black line) and the CM applied to the reconstructed normal distribution (blue dashed line) across various date selections and two stock choices. Notably, we observe a tendency for the curves to diverge, suggesting that the proposed method is structurally more robust. However, there are scenarios where sudden performance drops occur (e.g., bottom row, second column), followed by a recovery phase where the proposed strategy outperforms the classical Markowitz approach. This behavior may stem from instability in the threshold $\bar \mu$, which could temporarily affect the performance before stabilizing in favor of the proposed method.}
\label{fig_result_portfolio}
\end{figure}

In order to understand why the CM approach outperforms the classical Markowitz strategy we decided to look at cumulative return plots that are presented in Figure~\ref{fig_result_portfolio} for exemplary (representative) datasets.
 A noticeable divergence in the curves indicates that the reconstructed distribution tends to be more robust. However, specific scenarios exhibit abrupt performance drops (e.g., bottom row, second column) followed by a recovery phase where the reconstructed strategy surpasses the classical Markowitz approach. This phenomenon may arise due to a temporary instability in the threshold parameter $c$, which initially affects performance, but stabilizes over time to favor the reconstructed strategy. Nevertheless, the simplified conditional method presented in this paper might be enhanced to account for such phenomena. 
 
 We consider this strategy as a preliminary introduction to a potential further analysis. We acknowledge inherent simplifications, such as survival bias and omission of transaction costs. More advanced models, such as CAPM or Black-Litterman; see \cite{ sharpe1964capm,black1992}, account for market conditions and expert opinions, providing deeper insights. Furthermore, our approach primarily clusters data in space; a robust system should also diagnose crisis environments and adapt its decisions accordingly. However, our analysis suggests that the 20/60/20 rule may have promising applications in portfolio optimization, paving the way for future developments in this field.

 \section{Conclusions}
 \label{s:conclusions}
 This study highlights the potential of the 20/60/20 rule as a valuable tool in quantitative finance. By segmenting data into top-performing, average-performing, and underperforming groups, we demonstrate its applicability across multiple financial domains, including risk management, portfolio optimization, and statistical modeling of financial returns.

First, we show that applying the 20/60/20 rule to stock market data enables efficient population clustering, which can improve the identification of distinct performance groups. This segmentation provides a structured approach to analyzing financial data, particularly in scenarios where market heterogeneity plays a critical role. Second, we introduce a novel metric for assessing tail heaviness, derived using conditional statistical methods. This metric offers an easy-to-implement tool for characterizing extreme risks in financial datasets, contributing to more robust risk management strategies. Third, we integrate the 20/60/20 framework into the classical Markowitz portfolio optimization model, demonstrating its ability to enhance asset allocation and improve overall portfolio performance. Our findings indicate that this segmentation approach balances risk and return by capturing extreme movements while leveraging the stability of core market returns.

Beyond its relevance to normal distributions, the 20/60/20 rule demonstrates potential applications in financial mathematics, offering a foundation for further research on its integration into advanced financial models. Our analysis shows that observations from the middle part of the distribution often exhibit behavior similar to that of a normal distribution, which can enhance models relying on the covariance matrix. Additionally, we explored the preliminary potential of a new normality test based on the relationship between central observations and the total distribution. Potential extensions include its use in time clustering methods, dynamic portfolio rebalancing strategies, and alternative risk assessment frameworks. While the rule simplifies real-world trading algorithms and introduces a structured yet flexible approach, tailoring strategies to specific data segments can significantly improve investment outcomes. However, further refinement and empirical validation across different asset classes and market conditions remain important directions for future work.

\section*{Acknowledgements}
Marcin Pitera and Agnieszka Wy\l{}oma\'{n}ska acknowledge support from the National Science Centre, Poland, via project 2024/53/B/HS4/00433. Damian Jelito acknowledges support from the National Science Centre, Poland, via project 2024/53/B/ST1/00703.

\begin{footnotesize}
\bibliographystyle{agsm}

\end{footnotesize}

\end{document}